\documentclass{article}

\usepackage[a4paper,top=100pt,bottom=100pt,inner=80pt,outer=80pt]{geometry}

\usepackage[numbers,sort&compress]{natbib}

\RequirePackage[utf8]{inputenc}
\RequirePackage[T1]{fontenc}
\RequirePackage{lmodern}

\usepackage{amsmath}
\usepackage{amsthm}
\usepackage{amsfonts}
\usepackage[capitalize,sort&compress]{cleveref}

\newtheorem{theorem}{Theorem}[section]
\newtheorem{lemma}[theorem]{Lemma}

\newtheorem{corollary}[theorem]{Corollary}

\newtheorem{observation}[theorem]{Observation}

\usepackage{booktabs}
\usepackage{enumitem}

\usepackage{tikz}
\pgfdeclarelayer{bg}
\pgfsetlayers{bg,main}
\usetikzlibrary{shapes.misc,arrows.meta,decorations.pathreplacing,calligraphy}

\tikzset{
    .../.tip={
        [sep=0pt 2]
        Butt Cap[] .
        Rectangle[length=0pt 2, width=0pt 1]
        Rectangle[length=0pt 2, width=0pt 1]
        Rectangle[length=0pt 2, width=0pt 1, sep=0pt]
    }
}

\newcommand{\penp}{$\text{P}=\text{NP}$}
\newcommand{\pnenp}{$\text{P}\neq\text{NP}$}
\newcommand{\eqcomma}{\;,}
\newcommand{\eqdot}{\;.}

\usepackage{calc}
\usepackage{algorithm}
\usepackage{algorithmic}
\algsetup{linenodelimiter=}

\title{Maximin Shares in Hereditary Set Systems}
\author{
    Halvard Hummel
    \vspace{0.5em}
    \\
    \normalsize
    Norwegian University of Science and Technology
    \vspace{-0.3em}
    \\
    \small
    \texttt{halvard.hummel@ntnu.no}
}
\date{}

\begin{document}

\maketitle

\begin{abstract}
    We consider the problem of fairly allocating a set of indivisible items
    under the criteria of the maximin share guarantee. Specifically, we study
    approximation of maximin share allocations under hereditary set system
    valuations, in which each valuation function is based on the independent
    sets of an underlying hereditary set systems.  Using a lone divider
    approach, we show the existence of $1/2$-approximate MMS allocations,
    improving on the $11/30$ guarantee of \citeauthor{Li:21}.  Moreover, we
    prove that ($2/3 + \epsilon$)-approximate MMS allocations do not always
    exist in this model for every $\epsilon > 0$, an improvement from the recent
    $3/4 + \epsilon$ result of \citeauthor{Li:23a}. Our existence proof is
    constructive, but does not directly yield a polynomial-time approximation
    algorithm. However, we show that a $2/5$-approximate MMS allocation can be
    found in polynomial time, given valuation oracles.  Finally, we show that
    our existence and approximation results transfer to a variety of problems
    within constrained fair allocation, improving on existing results in some of
    these settings.
\end{abstract}

\section{Introduction}

The problem of fairly dividing a shared resource among a set of agents appears
frequently in real-world settings; Heritances need to be divided among heirs,
spots in university courses among students and office space among research
groups. This problem has a long history in economics, originating in the seminal
work of \citet{Steinhaus:48}. There, the shared resource is generally assumed to
be divisible, and the golden standard of fairness has long been
\emph{envy-freeness} and \emph{proportionality}. A division of a resource is
envy-free if no agent prefers the piece of another agent to her own. While
envy-free divisions exist, they have been notoriously hard to compute
efficiently \cite{Brandt:16}. Proportionality, which is easier to satisfy,
requires that each agent receives a piece worth, in her opinion, at least her
fair share of the total value, i.e., at least $1/n$ of the total value, where
$n$ is the number of agents.

In recent years, a variety of the fair division problem, in which the shared
resource is a set of indivisible items, has been extensively studied
\cite{Amanatidis:23}. For indivisible items, neither envy-freeness nor
proportionality is achievable in general; Consider, for example, allocating a
single item to two agents. No matter which agent receives the item, the other
will both be envious and receive less than her fair share of the items. Instead,
relaxed versions of these fairness criteria are considered. One of these---a
relaxation of proportionality---is the \emph{maximin share (MMS) guarantee},
introduced by \citet{Budish:11}. Informally, the MMS of an agent is the maximum
value she can guarantee herself if she had to partition the items, and got to
choose her own bundle last. An allocation satisfies the MMS guarantee if each
agent receives a bundle worth at least her own MMS. Maximin shares can be seen
as an extension of the famous \emph{cut-and-choose} protocol, in which one agent
splits the shared resource into two equal pieces and the other agent gets to
choose her favourite piece among the two and the agent that cut receives the
remaining piece \cite{Brandt:16}.

The existence and computation of MMS allocations has been considered in a range
of papers when valuations are additive \cite[see,
e.g.,][]{Aigner-Horev:22,Akrami:23,Akrami:24,Amanatidis:17a,Bouveret:16,Feige:22b,Garg:18a,Garg:20c,Ghodsi:21,Procaccia:14}.
Perhaps surprisingly, an MMS allocation is not guaranteed to exist with additive
valuations \cite{Procaccia:14,Bouveret:16}. In fact, it is sometimes not
possible to provide every agent with more than $39/40$ of her MMS
\cite{Feige:22b}. Although, there always exists allocations that provide each
agent with slightly more than $3/4$ of her MMS \cite{Akrami:24}, and there
exists polynomial-time algorithms that find allocations that provide each agent
with at least $3/4$ of her MMS \cite{Garg:20c}.

Additive valuations have so far been the main focus in the literature on fair
allocation of indivisible items. While additive valuations have a range of
useful properties, making it easier to find allocations that satisfy or to a
high-degree approximate various fairness criteria, they fall short in
expressiveness. Additive valuations can not express interactions between items,
such as expressing that items compliment or substitute each other. Thus, a
variety of recent papers have focused on more general classes of valuations,
considering both MMS and other fairness criteria \cite[see,
e.g.,][]{Ghodsi:22,Seddighin:24,Akrami:23e,Barman:20a,Uziahu:23,Kulkarni:23a,Chekuri:24}.

In parallel, a variety of fair allocation of indivisible items, known as
\emph{constrained fair allocation}, has garnered interest (see, e.g., the recent
surveys of \citet{Suksompong:21} and \citet{Biswas:23a}). In this version of the
problem, instances are furnished by a set of constraints on which combinations
of items can be given to the same agent, modelling restrictions found in a range
of real-world scenarios. For example, each agent may be endowed with a budget
and each item with a cost, and may only receive a collection of items with a
combined cost not exceeding her budget \cite{Barman:23e,Li:23a}. Alternatively,
the items may be partitioned into categories, and each agent may not receive
more items from a category than some given threshold
\cite{Biswas:18,Hummel:22a}.  Common for many constraints is that they can be
modelled through the valuations of the agents, by letting the value of a
collection items be the value of the maximum-valued subset that satisfies the
constraint.

In this paper, we consider approximation of MMS for a class of valuations that
models a number of different constraint types. Specifically, many of the
studied constraints are \textit{hereditary}, i.e., if a collection of items
satisfies the constraint, then every subset also satisfies the constraint.  The
class of valuations considered here, known as \emph{hereditary set system
valuations} and introduced by \citet{Li:21}, is based on the same principle. An
instance with hereditary set system valuations is given by a \emph{set system},
$H = (M, \mathcal{F})$, over the set of items $M$, that satisfies the
\emph{hereditary property}:
\[
    S \subseteq T \subseteq M \text{ and } T \in \mathcal{F} \implies S \in
    \mathcal{F}
\]
Each agent assigns to each item a value and the value of a collection of items,
$B$, is given by the sum of item values in the maximal-value subset of $B$ that
appears in $\mathcal{F}$. \citeauthor{Li:21} studied approximation of MMS under
this type of valuations, showing that there always exists an allocation
guaranteeing each agent at least $11/30$ of her MMS. Moreover, they showed that
such an allocation can be found in polynomial time if the valuations can be
queried in polynomial time.

\subsection{Contributions}

The main contribution of this paper is improvements to the lower and upper
bounds for existence and polynomial-time approximation of $\alpha$-approximate
MMS allocations under hereditary set system valuations. \Cref{tab:results}
provides an overview of the existence and approximation results. Specifically,
we show in \cref{sec:existence} that a $1/2$-approximate MMS allocation always
exists, improving on the $11/30$ existence guarantee of \citet{Li:21}. In fact,
our proof yields the slightly stronger existence guarantee of $n/(2n - 1)$,
where $n$ is the number of agents.  Moreover, we show that there are for any
number of agents, $n \ge 2$, instances for which no ($2/3 +
\epsilon$)-approximate MMS allocation exists for any $\epsilon > 0$, improving
on the $3/4 + \epsilon$ result of \citet{Li:23a}. Note that this fully decides
the case of $n = 2$, as $n/(2n - 1) = 2/3$ for $n = 2$.

The existence proof is constructive, making use of a lone divider style
algorithm \cite{Aigner-Horev:22}. Unfortunately, the proof requires the
computation of MMS partitions, which are NP-hard to approximate beyond a factor
of $2/3$ for hereditary set system valuations (see \cref{thr:2/3-np-hard} and
\cite{Li:23a}). In \cref{sec:approximation} we address this issue, showing that
a variation of the algorithm can produce $2/5$-approximate MMS allocations in
polynomial time, given a polynomial-time valuation oracle. In fact, we show that
the result holds even when supplied with an oracle that given a bundle $B$,
returns a subset $B' \subseteq B$, with $B' \in \mathcal{F}$, such that the
value of $B'$ is at least a $(1 - 1/(n + 1))$ fraction of the value of $B$. This
$2/5$-approximation algorithm improves on the $11/30$-approximation algorithm of
\citeauthor{Li:21}.

\begin{table}
    \centering
    \begin{tabular}{rll}
        \toprule
        & \textbf{Lower Bound} & \textbf{Upper Bound} \\
        \midrule
        \textbf{Existence} & $1/2$ (\Cref{thr:existence})& $2/3$
        (\Cref{thr:upper-bound}) \\
        \textbf{Approximation} & $2/5$ (\Cref{thr:2/5-poly}) & $2/3$
        (\citet{Li:23a}) \\
        \bottomrule
    \end{tabular}
    \caption{Existence and polynomial-time approximation guarantees for
    hereditary set system valuations. The polynomial-time approximation results
    assume polynomial-time valuation oracles.}
    \label{tab:results}
\end{table}

In \cref{sec:entitlement}, we consider natural extensions to hereditary set
system valuations, relaxing the requirement that the valuations of the agents
have to be based on the same hereditary set system. We show that the existence
and approximation results extend to the case where each agent has a valuation
function based on a different, but similar, hereditary set system.
Specifically, we require there to be an ordering of the agents, $a_1, \dots,
a_n$, such that for $i < j$, the hereditary set system for agent $a_j$'s
valuation function contains all independent set of the hereditary set system for
agent $a_i$'s valuation function. This setting naturally occurs in real-world
settings, such as in the earlier example with budgets. If the hereditary set
systems do not adhere to this requirement, we show that for any number of
agents, $n \ge 2$, there exists instances for which $1/2$ is the best possible
approximation.

Finally, in \cref{sec:constraints} we consider the impacts of our results on
constrained fair allocation. We show that for several types of constraints, our
results improve current existence or approximation results. In particular, we
show that for budget constraints our results improve the existence bound from
$1/3$ to $1/2$ and produce a $2/5$-approximation algorithm, improving on the
existing ($1/3 - \epsilon$)-approximation algorithm \cite{Li:23a}. For
conflicting items constraints \cite{Hummel:22}, we also improve the existence
bound from $1/3$ to $1/2$. Finally, for interval scheduling constraints, we
improve the existence bound from $1/3$ to $1/2$ and produce a
$2/7$-approximation algorithm, improving on the existing $(0.24 -
\epsilon)$-approximation algorithm \cite{Li:21f}.

\subsection{Additional Related Work}

A range of papers have considered existence and computation of MMS for more
general classes of valuation functions. Of particular interest for hereditary
set system valuations are the results for fractional subadditive (XOS)
valuations, as hereditary set system valuations are a subclass of XOS
valuations. \citet{Ghodsi:22} showed that $1/5$-approximate MMS allocations
always exist for XOS valuations, and presented a $1/8$-approximation algorithm.
They also showed that there exists instances with XOS valuations for which no
allocation is more than $1/2$-approximate MMS. Their existence result was later
improved to $0.219225$ by \citet{Seddighin:24} and more recently $3/13$ by
\citet{Akrami:23e}. \citeauthor{Akrami:23e} also provided a non-polynomial,
randomised algorithm that yields a $1/4$ approximation in expectation. MMS has
also been considered for valuations that are submodular
\cite{Barman:20a,Ghodsi:22,Uziahu:23}, subadditive
\cite{Ghodsi:22,Seddighin:24}, OXS \cite{Kulkarni:23a} and SPLC
\cite{Chekuri:24}.

A range of hereditary constraints have been considered. \citet{Biswas:18}
introduced \emph{cardinality constraints}, in which the items are partitioned
into categories, each with an attached limit on the number of items an agent can
receive from the category. They showed, among other results, that a
$1/3$-approximate MMS allocation always exists and can be found in polynomial
time under cardinality constraints. This was later improved to $1/2$ by
\citet{Hummel:22a}. Cardinality constraints are special case of \emph{matroid
constraints}. Under matroid constraints, each bundle must form an independent
set in a given matroid, which for cardinality constraints is a partition
matroid. Except for the work of \citet{Gourves:19}, which considered a seemingly
similar, but different model in which the combined set of items allocated to all
agents must be independent in the matroid, MMS has not specifically been
considered for general matroid constraints. Although, other fairness criteria
have been studied under matroid constraints \cite{Biswas:18,Biswas:19,Dror:23}.

Another type of hereditary constraints is the \emph{budget constraints}
introduced by \citet{Wu:21}. Budget constraints are equivalent to the earlier
described constraints with item costs and agent budgets. Recently,
\citet{Li:23a} considered MMS under budget constraints, showing the existence
and polynomial time findability of $1/3$-approximate MMS allocations ($1/3 -
\epsilon$ in polynomial-time). Moreover, they showed that for some instances no
allocation provides every agent with more than $3/4$ of her MMS, and that MMS
cannot in polynomial time be approximated within a factor greater than $2/3$,
unless \penp.

\citet{Chiarelli:23a} considered fair allocation of \emph{conflicting items}. In
this model, the items are given as vertices in a graph, and an agent can only
receive a set of items that forms an independent set in the graph.
\citeauthor{Chiarelli:23a} considered the complexity of the closely related
problem of finding an allocation that maximises the value the worst-off agent
receives. This problem is equivalent to deciding the MMS of an agent.
\citet{Hummel:22} considered MMS under conflicting items, showing the existence
of $1/3$-approximate MMS allocations, and provided a $1/\Delta$-approximation
algorithm, where $\Delta$ is the maximum degree of the graph. \citet{Li:21f}
considered the similar setting of interval scheduling constraints, in which each
item represents an interval and is endowed with a duration, a release time and a
deadline.\footnote{The model assumes discreet time steps.} The value of a bundle
for an agent is decided by the maximum-weight subset of items that can be
scheduled so that every item is started no earlier than their release time,
completed before their deadline and no pair of items are scheduled at the same
time. \citeauthor{Li:21f} proved the existence of $1/3$-approximate MMS
allocations and the polynomial-time findability of $(0.24 -
\epsilon)$-approximate MMS allocations in this model.

MMS has also been considered for constraints that are not hereditary, such as
the \emph{connectivity constraints} of \citet{Bouveret:17}.

\section{Preliminaries}

An instance of the fair allocation problem is given by a set of \emph{agents},
$N = \{1, \dots, n\}$, a set of \emph{items}, $M = \{1, \dots, m\}$, and
collection $V$ of \emph{valuation functions}, $v_i : 2^M \to \mathbb{R}_{\ge
0}$, one for agent $i \in N$. The goal is to find an $n$-partition of the items,
$A = (A_1, \dots, A_n)$, called an \emph{allocation}, that assigns to agent $i
\in N$ the items in $A_i$. We call $A_i$ and more generally any set $B \subseteq
M$ of items a \emph{bundle}. An $n$-partition of a subset of items $M' \subset
M$ is called a \emph{partial allocation}. An allocation that is not partial will
sometimes be referred to as a \emph{complete allocation}.

More than finding an allocation, the goal of the fair allocation problem is to
find an allocation that satisfies some fairness criteria. In this paper, we
consider allocations that satisfy the \emph{maximin share guarantee}, introduced
by \citet{Budish:11}. Formally, we define the \emph{maximin share (MMS)},
$\mu_i^n$, of an agent, $i$,
as
\[\mu_i^n = \max_{A \in \Pi_n(M)} \min_{A_j \in A} v_i(A_j)\eqcomma\]
where $\Pi_n(M)$ is the set of all $n$-partitions of $M$. If $n$ is obvious from
context, we use simply $\mu_i$.
Any $n$-partition, $P = (P_1, \dots, P_n)$, with $v_i(P_j) \ge \mu_i$ for every
$P_j \in P$ is called an \emph{MMS partition} of $i$. There always exists at
least one MMS partition of an agent $i$, since any partition maximising the
expression defining $\mu_i$ is an MMS partition. Given an $\alpha \le 1$, we say
that an allocation, $A = (A_1, \dots, A_n)$, is an \emph{$\alpha$-approximate
MMS allocation} if $v_i(A_i) \ge \alpha\mu_i$ for every agent $i \in N$. If
$\alpha = 1$, then $A$ is an \emph{MMS allocation}. Note that computing the MMS
of an agent is NP-hard, even in the case of two agents with additive valuations,
due to a reduction from PARTITION.

\subsection{Hereditary Set System Valuations}

We restrict our consideration to problem instances with hereditary set system
valuations. A \emph{hereditary set system}, $H = (J, \mathcal{F})$, is a family
$\mathcal{F}$ of subsets of a set $J$, such that for any pair of sets $S
\subseteq T \subseteq J$ with $T \in \mathcal{F}$, it holds that $S \in
\mathcal{F}$. We say that a subset $T \subseteq J$ is independent in $H$ if $T
\in \mathcal{F}$. An instance has \emph{hereditary set system valuations} if
there exists a hereditary set system $H = (M, \mathcal{F})$, such that $v_i(B) =
\max_{S \in \mathcal{F}} \sum_{j \in S \cap B} v_{ij}$. That is, each agent $i$
assigns to each item $j$ a value $v_{ij}$ and the value of a bundle $B$ is given
by the additive value of the maximum-value subset of $B$ independent in $H$. We
say that such a valuation function, $v_i$, is \emph{based on} $H$. If $v_i$ is
based on $H$, then a bundle $B$ is independent for $v_i$ if $B \in \mathcal{F}$.
A valuation function has \emph{binary marginal gains} if $v_i(B \cup \{j\}) -
v_i(B) \in \{0, 1\}$ for any $B \subseteq M$ and $j \in M \setminus B$. Notice
that a valuation function $v_i$ based on $H$ has binary marginal gains if an
only if $v_{ij} \in \{0, 1\}$ for every $j \in M$. Additionally, note that hereditary
set system valuations are monotone. Thus, if there exists a partial allocation
granting each agent a bundle worth at least her MMS, there exists an MMS
allocation that can be obtained by allocating the remaining items to an
arbitrary agent.

A valuation function, $v$, is \emph{additive} if $v(S) = \sum_{j \in S}
v(\{j\})$ and \emph{fractionally subadditive (XOS)} if there exists a finite
number of additive valuation functions, $f_1, \dots, f_k$, such that $v(S) =
\max_{j} f_j(S)$.  Note that hereditary set system valuations are a
subclass of XOS valuations, as an additive valuation function $f$ can be
constructed for each independent bundle of the hereditary set
system.\footnote{It is sufficient to create a valuation function, $f$, for each
maximal independent bundle.} As was pointed out by \citet{Li:21}, hereditary set
system valuations are not a subclass of the more restrictive \emph{submodular}
valuation functions.

Before continuing, we note a simple, but useful property of hereditary set
system valuations.

\begin{lemma}\label{lem:independent-equal-value}
    Given a bundle $B$, a hereditary set system $H = (M, \mathcal{F})$ and a
    valuation function $v$ based on $H$, there exists an independent bundle $B'
    \subseteq B$, with $v(B') = v(B)$. If $v$ can be queried in polynomial time,
    then $B'$ can be found in polynomial time.
\end{lemma}

\begin{proof}
    If $B$ is already independent, setting $B' = B$ is sufficient.  Otherwise,
    let $S \in \mathcal{F}$ be the independent set of $H$ that maximises
    $\sum_{j \in S \cap B} v_{ij}$. By definition, $v(S \cap B) = v(B)$.
    Moreover, $S \cap B$ must be an independent set of $H$, and $B' = S \cap B$
    satisfies the requirements.

    If $v$ can be queried in polynomial time, then $B'$ can be obtained by the
    following greedy algorithm, which can trivially be shown to run in
    polynomial time, with $O(|B|) = O(m)$ queries to the valuation function.
    \begin{algorithmic}[1]
        \STATE $B' = B$
        \FOR{$j \in B$}
            \IF{$v(B' \setminus \{j\}) = v(B')$}
                \STATE $B' = B' \setminus \{j\}$
            \ENDIF
        \ENDFOR
    \end{algorithmic}
    The value of $B'$ never decreases, thus it must hold that $v(B') = v(B)$
    when the algorithm completes. Moreover, there exists no good $j \in B'$ such
    that $v(B' \setminus \{j\}) = v(B')$. If $B'$ was not independent, then
    there would exist an $S \in \mathcal{F}$ such that $S \cap B' \neq B'$ and
    $v(B') = \sum_{j \in S \cap B'} v_{ij}$. For any good $j \in (B' \setminus
    S)$, of which there would be at least one, it would hold that $v(B'
    \setminus \{j\}) = v(B')$, a contradiction. Thus, $B'$ must be independent.
\end{proof}

\Cref{lem:independent-equal-value} is of particular importance to our results,
as independent bundles have the useful property that their value is equal to the
sum of the individual items within the bundle. Thus, allocating independent
bundles hinders wastage of items.

The size of the underlying hereditary set system may be exponential in the
number of items. Moreover, it can be shown that the number of distinct
hereditary set systems is very large (see, e.g., \cite{Li:21} for details).
Thus, it is customary to make the assumption of polynomial-time valuations
oracles when considering polynomial-time algorithms. The standard assumption is
an oracle that given a bundle $B$ returns the value of the bundle in polynomial
time. This is the type of oracle used in the $11/30$-approximation algorithm of
\citeauthor{Li:21}. In order to make our $2/5$-approximation algorithm
applicable to a larger range of hereditary set systems, we consider in addition
a weaker type of valuation oracle. Specifically, we have the following two types
of oracles:

\begin{description}
    \item[Exact valuation oracle] --- Returns in polynomial time for any bundle
        $B \subseteq M$ and agent $i \in N$, the value $v_i(B)$.
    \item[Approximate valuation oracle] --- Returns in polynomial time for any
        bundle $B \subseteq M$ and agent $i \in N$, an independent bundle $B'
        \subseteq B$ with $v_i(B') \ge (1 - \epsilon)v_i(B)$ for some specified
        error bound $0 \le \epsilon < 1$.
\end{description}

Note that an approximate valuation oracle, with error bound $\epsilon = 0$, can
be constructed from an exact valuation oracle by
\cref{lem:independent-equal-value}. Thus, any result that holds for approximate
valuation oracles with error bound $\epsilon \ge 0$, also holds for exact
valuation oracles. Moreover, note that if a bundle $B$ is known to be
independent, then the value of $B$ can be computed in polynomial time, even with
approximate valuation oracles. Additionally, the independence of a singleton
bundle $\{j\}$ can checked in polynomial time by an approximate valuation oracle
$v_i^o$, if $v_{ij} > 0$. This follows from the fact that $v_i(v_i^o(\{j\})) >
(1 - \epsilon)v_i(\{j\}) > 0$ if $v_i(\{j\}) > 0$.

An example of a use case in which an approximate valuation
oracle exists, but an exact valuation oracle does not (assuming \pnenp), is the
case in which calculating the value of a bundle is NP-hard, but there exists a
FPTAS that solves this task. This is the case for budget constraints, in which
finding the value of a bundle is equivalent to solving an instance of the
NP-hard knapsack problem.

\subsection{Lone Divider}

The improvements to existence and approximation guarantees make use of the lone
divider algorithm of \citet{Aigner-Horev:22}. Their algorithm, given as
\cref{alg:lone-divider}, takes as input a fair allocation instance and a
threshold, $x_i$, for each agent $i$ in the instance. The algorithm attempts to
produce a, potentially partial, allocation $A$, such that each agent receives a
bundle worth at least her threshold, i.e., $v_i(A_i) \ge x_i$. Note that if $x_i
= \alpha\mu_i$ for some $\alpha > 0$ and every agent $i$, the resulting
allocation, if any, is $\alpha$-approximate MMS.

\begin{algorithm}[t]
    \caption{Lone divider algorithm of \citeauthor{Aigner-Horev:22} (Algorithm 4
    in~\cite{Aigner-Horev:22})}
    \label{alg:lone-divider}
    \begin{algorithmic}[1]
        \REQUIRE A set of agents, $N$, a set of items, $M$, valuation functions,
        $V$, and values $x_i$ for $i \in N$.
        \ENSURE A (partial) allocation $A$ with $v_i(A_i) \ge x_i$
        \WHILE{$N \neq \emptyset$}
        \STATE select some $i \in N$\label{alg:line:divider}
        \STATE create $|N|$ pairwise disjoint bundles, $B_1, \dots, B_{|N|}$,
        with $v_i(B_j) \ge x_i$ for every $B_j$\label{alg:line:bundles}
        \STATE find a maximal-cardinality envy-free matching, $M_{EF}$, with
        regards to $N$, in the bipartite graph $G = (N \cup (B_1, \dots,
        B_{|N|}), \{(i', B_j) \in N \times (B_1, \dots, B_{|N|}) : v_{i'}(B_j)
        \ge x_{i'}\})$\label{alg:line:envy-free}
        \STATE allocate to each matched agent $i$ in $M_{EF}$ their matched
        bundle in $M_{EF}$
        \STATE update $N$ and $M$ by removing matched agents and items from
        matched bundles in $M_{EF}$
        \ENDWHILE
    \end{algorithmic}%
\end{algorithm}

The algorithm attempts to find a satisfactory allocation by iteratively
allocating satisfactory bundles to subsets of agents, such that the value of
each allocated bundle is small for the remaining agents. That is, for any
allocated bundle $B$, it holds that $v_{i'}(B) < x_{i'}$ for any remaining agent
$i'$. In order to achieve this, the algorithm selects in each iteration some
remaining agent $i$, the \emph{lone divider}, which is tasked with constructing
an equivalent number of bundles as there are remaining agents. Each of the
constructed bundles, $B_j$, must have a value of at least $x_i$ to agent $i$,
i.e., $v_i(B_j) \ge x_i$. The algorithm then considers the bipartite graph given
by the remaining agents (including $i$) and the constructed bundles, with edges
between any agent $i'$ and bundle $B_j$ where $v_{i'}(B_j) \ge x_{i'}$. An
envy-free matching with regards to the agents is found in the graph, and matched
agents are allocated their matched bundle. An \emph{envy-free matching with
regards to the agents} is here a matching in which every non-matched agent is
not connected by an edge to any matched bundle. This guarantees that for every
allocated bundle $B$ and non-matched agent $i'$, $v_{i'}(B) < x_{i'}$.

\citeauthor{Aigner-Horev:22} proved a sufficient condition for the existence of
non-empty envy-free matchings, similar to Hall's marriage theorem for matchings
in bipartite graphs.

\begin{lemma}[Corollary 1.1
    in~\cite{Aigner-Horev:22}]\label{lem:envy-free-matching}
    A bipartite graph $G = (X \cup Y, E)$ admits a non-empty envy-free matching
    with regards to $X$ if $|N_G(X)| \ge |X| > 0$, where $N_G(X)$ is the
    neighbourhood of $X$ in $G$.
\end{lemma}

Since the lone divider, $i$, is tasked with creating an equal number of bundles
as there are remaining agents, each of which will be connected to $i$ in the
bipartite graph, a non-empty envy-free matching exists by
\cref{lem:envy-free-matching}. Thus, the number of agents decreases by at least
one. Moreover, since $i$ is connected to every bundle by an edge, $i$ is always
part of the matching. \citeauthor{Aigner-Horev:22} showed through this argument
that the lone divider algorithm finds a satisfactory allocation if the selected
lone divider is always able to construct a sufficient number of bundles.

\begin{theorem}[Theorem~4.1 in~\cite{Aigner-Horev:22}]\label{thr:sufficient-conditions}
    \Cref{alg:lone-divider} finds a (partial) allocation $A$ with $v_i(A_i) \ge
    x_i$ for each agent $i$, if thresholds $x_i$ are chosen such that
    \begin{enumerate}[label=(\arabic*)]
        \item there exists a set of pairwise disjoint bundles $B_1, \dots, B_n$
            with $v_i(B_j) \ge x_i$ for every $j \in \{1, \dots, n\}$;
            and\label{item:base-requirement}
        \item for every $k \in \{1, \dots, n - 1\}$ and any set of $k$ bundles,
            $C_1, \dots, C_k$, that can be allocated before $i$ is selected as
            the lone divider, there exists a set of pairwise disjoint bundles
            $B_1, \dots, B_{n - k}$ with $v_i(B_j) \ge x_i$ for every $j \in
            \{1, \dots, n - k\}$, and $\bigcup_{1 \le j \le n - k} B_j \subseteq
            (M \setminus \bigcup_{1 \le j \le k}
            C_j)$.\label{item:continued-requirement}
    \end{enumerate}
\end{theorem}

Further, as is pointed out by \citeauthor{Aigner-Horev:22}, if valuations
functions can be queried in polynomial time and the creation of the bundles
(line~\ref{alg:line:bundles}) can be performed in polynomial time, the algorithm
runs in polynomial time.  Specifically, this follows from the fact that a
maximal-cardinality envy-free matchings can be computed in polynomial time.

\begin{theorem}[Theorem 1.2 in~\cite{Aigner-Horev:22}]
    A maximal-cardinality envy-free matching can be found in polynomial time in
    the number of vertices and edges in the graph.
\end{theorem}

Finally, note that the concept of lone divider algorithms predates the work of
Aigner-Horev and Segal-Halevi, originating in the cake-cutting literature
\cite{Steinhaus:48,Kuhn:67}. Moreover, the lone divider algorithm discussed here
has recently been used in existence and approximation algorithms for MMS under
fair allocation of indivisible items, including in the paper of
\citeauthor{Aigner-Horev:22}
\cite{Hummel:22a,Aigner-Horev:22,Hosseini:21,Hosseini:22d}.

\section{Existence}\label{sec:existence}

In this section, we improve both the lower and upper bound guarantees for the
existence of $\alpha$-approximate MMS allocations under hereditary set system
valuations. First, we consider the improvement to the lower bound. In order to
obtain the given existence guarantee of $1/2$, we will consider a slight
variation of \cref{alg:lone-divider}. Specifically, we require each bundle
created on line~\ref{alg:line:bundles} to be independent in the underlying
hereditary set system. Recall that by \cref{lem:independent-equal-value}, the
requirement of constructing enough independent bundles worth at least $\mu_i/2$
is not a stricter requirement than simply constructing enough bundles worth at
least $\mu_i/2$.

The requirement of bundles that are independent in the hereditary set system is
made to limit the amount of lost value from each allocated bundle for the
remaining agents. If a bundle is not independent, then the bundle contains
wasted items, i.e., items that do not make a contribution to the bundle's value.
In other words, the sum of the individual values for the items in the bundle may
greatly exceed the value of the bundle. Thus, unless the bundles are
independent, the envy-free matching does not restrict the lost item value for
each allocated bundle to $x_i$. This could make it impossible for later agents
to create enough satisfactory bundles. With independent bundles, it holds by the
definition of hereditary set system valuations that $v_{i'}(B) = \sum_{j \in B}
v_{i'j}$ for any agent $i'$. Thus, the value of the allocated bundles accurately
bounds the value of the removed items.

\begin{theorem}\label{thr:existence}
    A $1/2$-approximate MMS allocation always exists under hereditary set system
    valuations.
\end{theorem}

\begin{proof}[Proof of \cref{thr:existence}]
    Consider a variety of \cref{alg:lone-divider} in which every bundle created
    on line~\ref{alg:line:bundles} is independent. It is sufficient to show that
    the two conditions of \cref{thr:sufficient-conditions} hold when $x_i =
    \mu_i/2$. Note that requirement~\ref{item:base-requirement} holds for any
    $x_i \le \mu_i$, as any MMS partition, $P = (P_1, \dots, P_n)$, of agent $i$
    contains $n$ bundles, each with $v_i(P_j) \ge \mu_i$.

    To show that requirement~\ref{item:continued-requirement} holds, fix some
    $k \in \{1, \dots, n - 1\}$ and set of independent bundles $C_1, \dots,
    C_{k}$ with $v_i(C_j) < x_i$ for every $j \in \{1, \dots, k\}$. Let $M' = M
    \setminus \bigcup_{1 \le j \le k} C_k$ be the set of unallocated items.
    Moreover, let $P = (P_1, \dots, P_n)$ be an MMS partition of agent $i$ and
    $P^* = \{P_j \cap M' : P_j \in P, v_i(P_j \cap M') \ge x_i\}$ be the set of
    bundles in $P$ valued at $x_i$ or greater, after the removal of the already
    allocated items. We wish to show that the number of bundles in $P^*$ is at
    least $n - k$, i.e., $|P^*| \ge n - k$. Then, any subset of $P^*$ with
    cardinality $n - k$ satisfies the requirements for $B_1, \dots, B_{n - k}$.

    Consider a bundle $P_j \in P$. We have that $v_i(P_j) \ge v_i(P_j \cap M')$
    by the monotonicity of the valuation functions. Moreover, if $v_i(P_j \cap
    M') < x_i$, it must hold that $v_i(P_j) - v_i(P_j \cap M') > (\mu_i - x_i)$.
    Since each $C_j$ is independent, it follows that
    \[
        \sum_{j = 1}^n v_i(P_j) - v_i(P_j \cap M') \le
        \sum_{j \in M \setminus M'} v_{ij} =
        \sum_{j = 1}^k v_i(C_j) < kx_i \eqdot
    \]
    Where the first step follows from the fact that each item $j$ contributes at
    most $v_{ij}$ to the value of a bundle $B$ with $j \in B$. For $|P^*| < n -
    k$, it must therefore hold that $kx_i > (k + 1)(\mu_i - x_i)$. In other
    words, if $kx_i \le (k + 1)(\mu_i - x_i)$, then $|P^*| \ge n - k$. Letting
    $x_i = \alpha\mu_i$ for some $\alpha > 0$, we get the following limitations
    on the value of $\alpha$ for which $|P^*| \ge n - k$ is guaranteed.
    \begin{align*}
        k\alpha\mu_i &\le (k + 1)(\mu_i - \alpha\mu_i) \\
        k\alpha &\le (k + 1) - (k\alpha + 1) \\
        (2k + 1)\alpha &\le k + 1 \\
        \alpha &\le \frac{k + 1}{2k + 1}
    \end{align*}
    Since $1/2 \le (k + 1)/(2k + 1)$ for any $k \ge 1$,
    requirement~\ref{item:continued-requirement} holds for $x_i = \mu_i/2$.
    The existence of $1/2$-approximate MMS allocations is therefore guaranteed
    by \cref{thr:sufficient-conditions}.
\end{proof}

The calculations in the proof of \cref{thr:existence} do not make use of the
fact that $k$ is bounded by $n - 1$ in \cref{thr:sufficient-conditions}.
Exploiting this fact yields a slightly better existence guarantee as a
corollary. Note that this improved bound tends to $1/2$ as the number of agents
increases, but is significantly larger when the number of agents is small.

\begin{corollary}\label{lem:existence}
    A ($n/(2n - 1)$)-approximate MMS allocation always exists under hereditary
    set system valuations.
\end{corollary}

\begin{proof}
    The proof of \cref{thr:existence} guarantees that
    requirement~\ref{item:base-requirement} holds as long as $\alpha \le 1$.
    Moreover, it shows that requirement~\ref{item:continued-requirement} holds
    when $\alpha \le (k + 1)/(2k + 1)$. As $(k + 1)/(2k + 1)$ is monotonically
    decreasing in $k$, the strongest requirement on $\alpha$ is given by the
    largest value of $k$, which is $k = n - 1$. Thus, the requirement holds for
    \[\alpha \le \frac{(n - 1) + 1}{2(n - 1) + 1} = \frac{n}{2n - 1}\eqdot\]
    By \cref{thr:sufficient-conditions}, a ($n/(2n - 1)$)-approximate MMS
    allocation exists.
\end{proof}

To complete our consideration of existence guarantees, we provide an upper bound
on existence of $\alpha$-approximate MMS allocations for every number of agents.
That is, for any number of agents, except the trivial case of one agent, there
exists an instance in which no approximation better than $2/3$ is possible.
Combined with \cref{lem:existence}, this fully determines the case of two
agents, where a $2/3$-approximate MMS allocation always exists, but no better
approximation is guaranteed to exist.

\begin{theorem}\label{thr:upper-bound}
    For any $n \ge 2$ and $\epsilon > 0$, there exists an instance with $n$
    agents with hereditary set system valuations with binary marginal gains,
    such that no $(2/3 + \epsilon)$-approximate MMS allocation exists.
\end{theorem}

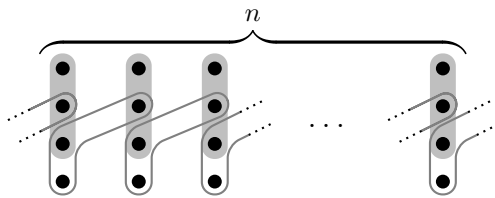
\begin{figure}
    \centering
    \begin{tikzpicture}
        \draw[every node/.append style={circle, fill=black, inner sep=0.65mm}]
            (0, 0) node (v11) {}
            (0, 0.5) node (v12) {}
            (0, 1) node (v13) {}
            (0, 1.5) node (v14) {}
            (1, 0) node (v21) {}
            (1, 0.5) node (v22) {}
            (1, 1) node (v23) {}
            (1, 1.5) node (v24) {}
            (2, 0) node (v31) {}
            (2, 0.5) node (v32) {}
            (2, 1) node (v33) {}
            (2, 1.5) node (v34) {}
            (5, 0) node (v41) {}
            (5, 0.5) node (v42) {}
            (5, 1) node (v43) {}
            (5, 1.5) node (v44) {}
            ;

        \node at (3.5, 0.75) {\large$\dots$};

        \draw[decorate, decoration={calligraphic brace,amplitude=3mm}, very thick]
        (-0.3, 1.7) -- (5.3, 1.7);
        \node at (2.5, 2.2) {$n$};

        \begin{pgfonlayer}{bg}
            \draw[draw=none,fill=gray!50,rounded corners=1.65mm]
            ([xshift=-1mm, yshift=1.3mm]v14.north west) rectangle
            ([xshift=1mm, yshift=-1.3mm]v12.south east)
            ([xshift=-1mm, yshift=1.3mm]v24.north west) rectangle
            ([xshift=1mm, yshift=-1.3mm]v22.south east)
            ([xshift=-1mm, yshift=1.3mm]v34.north west) rectangle
            ([xshift=1mm, yshift=-1.3mm]v32.south east)
            ([xshift=-1mm, yshift=1.3mm]v44.north west) rectangle
            ([xshift=1mm, yshift=-1.3mm]v42.south east)
            ;

            \draw[thick, draw=gray, rounded corners=1.65mm]
                ([xshift=-1mm, yshift=-1mm]v11.south west) --
                ([xshift=-1mm, yshift=1mm]v12.north west) --
                ([xshift=0.5mm, yshift=1.575mm]v23.north east) --
                ([xshift=1.5mm, yshift=-0.075mm]v23.south east) --
                ([xshift=1mm, yshift=0.5mm]v12.south east) --
                ([xshift=1mm, yshift=-1mm]v11.south east) -- cycle;

            \draw[thick, draw=gray, rounded corners=1.65mm]
                ([xshift=-1mm, yshift=-1mm]v21.south west) --
                ([xshift=-1mm, yshift=1mm]v22.north west) --
                ([xshift=0.5mm, yshift=1.575mm]v33.north east) --
                ([xshift=1.5mm, yshift=-0.075mm]v33.south east) --
                ([xshift=1mm, yshift=0.5mm]v22.south east) --
                ([xshift=1mm, yshift=-1mm]v21.south east) -- cycle;

            \draw[thick, draw=gray, rounded corners=1.65mm]
                ([xshift=1mm, yshift=-2.5mm]v32.south east) --
                ([xshift=1mm, yshift=-1mm]v31.south east) --
                ([xshift=-1mm, yshift=-1mm]v31.south west) --
                ([xshift=-1mm, yshift=1mm]v32.north west) --
                ++(0.5, 0.25)
                edge[dotted] ++(0.3, 0.15)
                ([xshift=1mm, yshift=-2.5mm]v32.south east) --
                ([xshift=1mm, yshift=0.5mm]v32.south east) --
                ++(0.275, 0.1375)
                edge[dotted] ++(0.3, 0.15)
                ;

            \draw[thick, draw=gray, rounded corners=1.65mm]
                ([xshift=1mm, yshift=-2.5mm]v42.south east) --
                ([xshift=1mm, yshift=-1mm]v41.south east) --
                ([xshift=-1mm, yshift=-1mm]v41.south west) --
                ([xshift=-1mm, yshift=1mm]v42.north west) --
                ++(0.5, 0.25)
                edge[dotted] ++(0.3, 0.15)
                ([xshift=1mm, yshift=-2.5mm]v42.south east) --
                ([xshift=1mm, yshift=0.5mm]v42.south east) --
                ++(0.275, 0.1375)
                edge[dotted] ++(0.3, 0.15)
                ;

            \draw[thick, draw=gray, rounded corners=1.65mm]
                ([xshift=-5.0mm, yshift=-1.125mm]v13.north east) --
                ([xshift=0.5mm, yshift=1.575mm]v13.north east) --
                ([xshift=1.5mm, yshift=-0.075mm]v13.south east) --
                ++(-0.5, -0.25)
                edge[dotted] ++(-0.3, -0.15)
                ([xshift=-3.5mm, yshift=-2.575mm]v13.south east) --
                ([xshift=1.5mm, yshift=-0.075mm]v13.south east) --
                ([xshift=0.5mm, yshift=1.575mm]v13.north east) --
                ++(-0.55,-0.27)
                edge[dotted] ++(-0.3, -0.1472)
                ;
            \draw[thick, draw=gray, rounded corners=1.65mm]
                ([xshift=-5.0mm, yshift=-1.125mm]v43.north east) --
                ([xshift=0.5mm, yshift=1.575mm]v43.north east) --
                ([xshift=1.5mm, yshift=-0.075mm]v43.south east) --
                ++(-0.5, -0.25)
                edge[dotted] ++(-0.3, -0.15)
                ([xshift=-3.5mm, yshift=-2.575mm]v43.south east) --
                ([xshift=1.5mm, yshift=-0.075mm]v43.south east) --
                ([xshift=0.5mm, yshift=1.575mm]v43.north east) --
                ++(-0.55,-0.27)
                edge[dotted] ++(-0.3, -0.1472)
                ;
        \end{pgfonlayer}
    \end{tikzpicture}
    \caption{An instance with no ($2/3 + \epsilon$)-approximate MMS allocation
    for any $\epsilon > 0$. The independent sets of the hereditary set system
    are given by the triples marked by the filled and outlined gray areas, and
    all their subsets. One agent assigns a value of 1 to each item in the three
    upper rows. The remaining $n - 1$ agents assign a value of 1 to each item in
    the three lower rows.}
    \label{fig:2/3-non-existence}
\end{figure}

\begin{proof}
    We consider an instance with $n$ agents and $4n$ items. Let $H = (M,
    \mathcal{F})$ be the hereditary set system given by the triples
    $\{4k+1,4k+2,4k+3\}$ and $\{4k+3,4k+4,4k+6\}$ for every $k \in \{0, 1,
    \dots, n - 1\}$, and all their subsets, where the item $4k + 6$ is replaced
    by item $2$ when $k = n - 1$. These triples are visualised in
    \cref{fig:2/3-non-existence}, where the filled gray areas are the triples
    $\{4k+1,4k+2,4k+3\}$ and the outlined gray areas are the triples $\{4k + 3,
    4k+4,4k+6\}$.

    We construct valuation functions, $v_1$ and $v_2$, based on $H$, with the
    following item values
    \[
        v_{1j} = \begin{cases}
            1 & j \neq 4k \text{ for every } k \in \mathbb{N} \\
            0 & \text{otherwise}
        \end{cases}
        \;\;\;\;\;\;\;\;\;\;\;
        v_{2j} = \begin{cases}
            1 & j \neq 4k + 1 \text{ for every } k \in \mathbb{N} \\
            0 & \text{otherwise}
        \end{cases}
    \]
    Using the visualisation of \cref{fig:2/3-non-existence}, $v_1$ is equivalent
    to assigning a value of $1$ to every item in the three upper rows and a
    value of $0$ to every item in the lower row. Equivalently, $v_2$ assigns a
    value of $1$ to every item in the three lower rows and a value of $0$ to
    every item in the upper row.

    The MMS of an agent with either valuation function is exactly $3$, as
    \[\mu_i \le (\sum_{j \in M} v_i(\{j\}))/n = 3n/n = 3\eqcomma\]
    for both $v_i = v_1$ and $v_i = v_2$. Moreover, there exists for both $v_1$
    and $v_2$ a partial $n$-partition where every bundle has value exactly $3$.
    For $v_1$ this is the partition $(\{1, 2, 3\}, \{5, 6, 7\}, \dots, \{4n - 3,
    4n - 2, 4n - 1\})$ that contains every triple $\{4k+1,4k+2,4k+3\}$ and for
    $v_2$ the partition $(\{3, 4, 6\}, \{7, 8, 10\}, \dots, \{4n - 2, 4n - 1,
    2\})$ that contains every triple $\{4k+3,4k+4,4k+6\}$. Note that for both
    $v_1$ and $v_2$, it can easily be verified that any bundle $B$ with $v(B)
    \ge 3$ contains a triple on the form $\{k+1,k+2,k+3\}$ and
    $\{k+3,k+4,k+6\}$, respectively.

    We finish constructing our instance, by letting one of the agents, $i$, have
    valuation function $v_1$ and the remaining $n - 1$ agents have valuation
    function $v_2$.\footnote{The proof works just as well with $r$ agents with
    valuation function $v_1$ and $n - r$ agents with valuation function $v_2$,
    as long as $0 < r < n$.} For any $\epsilon > 0$, a $(2/3 +
    \epsilon)$-approximate MMS allocation requires each agent, $i'$, to receive
    a bundle with value at least $3$, as $v_{i'}(B) \in \mathbb{N}$ and $\lceil
    2\mu_{i'}/3\rceil = 3$. Any bundle with value $3$ to agent $i$ overlaps with
    two of the $n$ bundles with value $3$ according to $v_2$ (every filled gray
    area in \cref{fig:2/3-non-existence} overlaps with two outlined gray areas).
    Thus, if agent $i$ receives a bundle with value at least $3$, then at most
    $n - 2$ of the remaining agents can receive a bundle with value at least
    $3$.  Consequently, it is impossible to guarantee every agent a bundle with
    value at least $3$ and no $(2/3 + \epsilon)$-approximate MMS allocation can
    exist.
\end{proof}

Note that the proof of \cref{thr:upper-bound} makes use of only two types of
agents. Thus, the nonexistence also holds when the types of agents is bounded.

\section{An Approximation Algorithm}\label{sec:approximation}

In the previous section, we saw that there always exists an $\alpha$-approximate
MMS allocation for $\alpha = 1/2$. While the proof is constructive, it requires
the computation of MMS partitions, one for each of the chosen lone dividers.
As computing MMS partitions is NP-hard, a polynomial-time algorithm cannot be
obtained directly from the proof. However, notice that if for some $\alpha$, a
$2\alpha$-approximate MMS partition could be found in polynomial time for each
agent, then the method used in the proof would yield a way to create $n - k$
bundles with value $\alpha\mu_i$ for the selected agent, $i$. Thus, the
existence of a good approximation algorithm for MMS partitions would yield an
all right approximation algorithm for MMS allocations.

In the additive case, a PTAS for $\alpha$-approximate MMS partitions exists,
permitting polynomial-time computation of $\alpha$-approximate MMS partitions
for any fixed choice of $\alpha < 1$ \cite{Woeginger:97}. If this algorithm
could be adapted to hereditary set system valuations, allocations that are
almost $1/2$-approximate MMS could be found in polynomial time. Unfortunately,
this is not possible for hereditary set system valuations. In fact, it can be
shown that it is impossible to compute $(2/3 + \epsilon)$-approximate MMS
partitions in polynomial time for any $\epsilon > 0$, unless \penp.

The inapproximation result follows directly from a result of \citet{Li:23a}.
They showed that there is exists no polynomial-time ($2/3 +
\epsilon$)-approximation algorithm for MMS under budget constraints, unless
\penp. While budget constraints is not a subclass of hereditary set system
valuations, their proof makes use of a case in which budget constraints overlap
with hereditary set system valuations. Specifically, the case in which all
agents are identical. For completeness and to highlight its validity in the
context of hereditary set systems, we give a variety of their proof using only
the context of hereditary set systems.

\begin{theorem}[Follows from Theorem~4 in~\cite{Li:23a}]\label{thr:2/3-np-hard}
    Given an $\epsilon > 0$, computing a $(2/3 + \epsilon)$-approximate MMS
    allocation, if one exists, is strongly NP-hard under hereditary set system
    valuations, even when agents are identical, have binary marginal gains and
    the independent sets of the hereditary set system are given as input.
\end{theorem}

\begin{proof}[Proof (Based on the proof of Theorem~4 in~\cite{Li:23a})]
    An instance of $3$-PARTITION is a multiset of $m$ positive integers, $a_1,
    \dots, a_m$. The problem asks if there exists a partition of $a_1, \dots,
    a_m$ into $n$ triples, with $m = 3n$, such that every triple sums to the
    same value, $T = (\sum_{j = 1}^{m} a_j)/n$. It has be shown that
    $3$-PARTITION is strongly NP-hard \cite{Garey:90}.

    Given an instance of $3$-PARTITION, we construct a hereditary set
    system $H = (\{1, \dots, m\}, \{S \subseteq \{1, \dots, m\} : |S| \le 3,
    \sum_{j \in S} a_j \le T\})$. As there are $O(m^3)$ subsets of cardinality
    at most three, $H$ can be constructed in polynomial time. Let $v$ be the
    valuation function given by $H$, with $v(j) = 1$ for $j \in \{1, \dots,
    m\}$. For any bundle $B \subseteq \{1, \dots, m\}$, $v(B) \ge 3$ holds if
    and only if there is a triple $\{j_1, j_2, j_3\} \subseteq B$ with $a_{j_1}
    + a_{j_2} + a_{j_3} \le T$. Thus, if and only if the answer to the
    $3$-PARTITION instance is YES, there exists an $n$-partition $(B_1,
    \dots, B_n)$ such that $v(B_j) \ge 3$ for every $B_j$ in the partition.

    Consider $n$ agents with valuation function $v$. It follows that $\mu_i = 3$
    for any agent $i$, if an only if the answer to the $3$-PARTITION instance is
    YES. Otherwise, as valuations are integer, it follows that $\mu_i \le 2$. In
    a similar fashion to in \cref{thr:upper-bound}, $\lceil(2/3 +
    \epsilon)\mu_i\rceil = \mu_i$ for any $\epsilon > 0$. Thus, any $(2/3 +
    \epsilon)$-approximate MMS allocation can be used to determine the answer to
    the original $3$-PARTITION instance in polynomial time. Since, the
    construction of the fair allocation problem from the $3$-PARTITION instance
    is polynomial, the strongly NP-hardness of $(2/3 + \epsilon)$-approximate
    MMS follows.
\end{proof}

As a consequence of \cref{thr:2/3-np-hard}, the approach used to created bundles
in the proof of \cref{thr:existence} can only work for $\alpha \le 1/3$ in
polynomial time (unless \penp). To achieve an approximation guarantee of $2/5$,
we instead create the required bundles directly from the set of remaining items
at each step of the lone divider algorithm. The algorithm for creating bundles,
\cref{alg:2/5-poly-partition}, works by first dividing the items into two
categories, $M_H$ and $M_L$, containing high- and low-valued items,
respectively. An item $j$ is considered high-valued to agent $i$, if $v_i(\{j\})
> \mu_i/5$. After constructing the two sets, bundles are created in two phases.
In the first phase, bundles are created that contain at most one item from
$M_H$. The construction is done greedily, considering each item $j \in M_H$ in
turn. For each item $j$, the approximate valuation oracle is queried for the set
$\{j\} \cup M_L$. If the returned bundle $B$ satisfies $v_i(B) \ge (2/5)\mu_i$,
the algorithm finds some bundle $B' \subseteq B$ such that $v_i(B') \ge
(2/5)\mu_i$ and for any $j' \in B'$, $v_i(B' \setminus \{j'\}) < (2/5)\mu_i$.
The simplification of $B$ to $B'$ can be done greedily, considering in turn
every item $j' \in B$. For each $j'$, it is checked if the item can be removed
from $B$ without the value dropping below $(2/5)\mu_i$. If this is the case, the
item is removed. After constructing the bundle $B'$, the process is continued
for the remaining items.

\begin{algorithm}[t]
    \caption{Find $k$ disjoint independent bundles with value at least
    $\frac{2}{5}\mu_i^*$ for valuation function $v_i$}
    \label{alg:2/5-poly-partition}
    \begin{algorithmic}[1]
        \REQUIRE An approximate valuation oracle $v^o_i$, items $M'$, a number $k$
        and an estimate $\mu_i^*$
        \ENSURE $k$ pairwise disjoint independent bundles $B_1, \dots, B_k$ with
        $v_i(B_j) \ge \frac{2}{5}\mu_i^*$ or \textsc{Failure}
        \STATE $M_H = \{j \in M' : v_i(\{j\}) > \frac{1}{5}\mu_i^*\}$
        \STATE $M_L = M' \setminus M_H$
        \COMMENT{Phase one}
        \FOR{$j \in M_H$ with $v_i(\{j\}) \ge \frac{2}{5}\mu_i^*$}
            \STATE create bundle $\{j\}$
            \STATE $M_H = M_H \setminus \{j\}$
        \ENDFOR
        \WHILE{$\exists j \in M_H$ with $v_i(v^o_i(\{j\} \cup M_L)) \ge
        \frac{2}{5}\mu_i^*$}\label{alg:line:while-1}
            \STATE $B = v^o_i(\{j\} \cup M_L)$
            \WHILE{$\exists j' \in B$ with $v_i(B \setminus \{j'\}) \ge
            \frac{2}{5}\mu_i^*$}
                \STATE $B = B \setminus \{j'\}$
            \ENDWHILE
            \STATE create bundle $B$
            \STATE $M_H = M_H \setminus B$
            \STATE $M_L = M_L \setminus B$
        \ENDWHILE
        \COMMENT{Phase two}
        \STATE $G = (M_H, \{\{j_1, j_2\} \in M_H \times M_H : v_i(v^o_i(\{j_1,
        j_2\})) \ge \frac{2}{5}\mu_i^*\})$
        \STATE Find a maximum-cardinality matching $M_G$ in $G$
        \FOR{$\{j_1, j_2\} \in M_G$}
            \STATE create bundle $\{j_1, j_2\}$
        \ENDFOR
        \IF{less than $k$ bundles have been created}
            \RETURN \textsc{Failure}
        \ELSE
            \RETURN $k$ created bundles
        \ENDIF
    \end{algorithmic}
\end{algorithm}

When the algorithm is no longer able to construct bundles in the first phase, it
moves on to the second phase. In this phase, it constructs as many bundles as
possible, such that for every constructed bundle $B$, it holds that $v_i(B) \ge
(2/5)\mu_i$ and $|B| = |B \cap M_H| = 2$. A maximum-cardinality collection of
pairwise disjoint bundles that satisfy these criteria can be found through a
maximum-cardinality matching in the graph $G = (M_H, \{\{j_1, j_2\} \in M_H
\times M_H : v_i(\{j_1, j_2\}) \ge (2/5)\mu_i\})$. Specifically, each edge in
the matching would correspond to a bundle comprised of the two items the edge
connects.

The idea behind the two phases is to initially restrict the value of each
allocated bundle, to avoid excessively overshooting the required value. In the
first phase, the value of a bundle is at most $(3/5)\mu_i$, unless the bundle is
a singleton. If the bundle contained more than one item and had a value
exceeding $(3/5)\mu_i$, there would be at least one item $j' \in M_L$ in the
bundle. Since $v_i(\{j'\}) \le \mu_i/5$, it would be possible to remove $j'$
from the bundle without reducing the value of the bundle by more than $\mu_i/5$.
As the value of the bundle would remain at least $(2/5)\mu_i$ after removing
$j'$, there can be no such item $j'$ and the value of the bundle is less than
$(3/5)\mu_i$.

Consider an MMS partition, $P = (P_1, \dots, P_n)$, of agent $i$, with all items
removed that have been allocated in the lone divider algorithm or used in the
bundles created in phase one. Since every one of these bundles, except the
singleton bundles from phase one, have a value of less than $(3/5)\mu_i$, at
most $r$ bundles can now have a value of less than $(2/5)\mu_i$, where $r$ is
the number of allocated and created bundles. Note that, the potentially higher
value of the singleton bundles does not affect this observation, as the removal
of a single item affects only the remaining value in a single bundle in $P$.
Thus, at least $n - r$ bundles have a remaining value of at least $(2/5)\mu_i$.
Moreover, every one of these $n - r$ bundles must contain at least two items
from $M_H$, as $v_i(\{j\} \cup M_L) < (2/5)\mu_i$ for every $j \in M_H$. As a
consequence, the matching in phase two will have cardinality of at least $n -
r$, resulting in a sufficient number of bundles.

While it may be clear that the described method will produce a sufficient number
of bundles, it assumes knowledge of $\mu_i$ and the argumentation does not make
any considerations in regards to the inaccuracy of the valuation oracle. As
computing $\mu_i$ is NP-hard, \cref{alg:2/5-poly-partition} makes instead use of
an estimate $\mu_i^*$. It can be shown that if this estimate is not too much
higher than $\mu_i$, the algorithm is guaranteed to find a satisfactory number
of bundles with value at least $(2/5)\mu_i^*$, even when supplied with an
approximate valuation oracle with error bound by $1/(n + 1)$.

\begin{lemma}\label{lem:2/5-partition-bounds}
    Consider an instance given by a set $N$ of $n \ge 2$ agents, a set $M$ of
    items and hereditary set system valuations $V$. Let $\mu^*_i > 0$ be an
    estimate of the MMS of an agent $i \in N$, with $\mu^*_i \le (1 + 1/(5n -
    1))\mu_i$, and $v^o_i$ an approximate valuation oracle for $v_i$, with error
    bound by $1/(n + 1)$. Given $0 \le \ell < n$ pairwise disjoint bundles $C_1,
    \dots, C_\ell$ with $\sum_{j' \in C_j} v_{ij'} \le (3/5)\mu^*_i$,
    \cref{alg:2/5-poly-partition} will return $k = n - \ell$ pairwise disjoint
    independent bundles $B_1, \dots, B_k$ with $v_i(B_j) \ge (2/5)\mu^*_i$, when
    ran for $v^o_i$, $M' = M \setminus (\bigcup_{1 \le j \le \ell}C_j)$, $k$ and
    $\mu^*_i$.
\end{lemma}

\begin{proof}
    We wish to show that the algorithm creates at least $k$ pairwise disjoint
    independent bundles $B_1, \dots B_k$, with $v_i(B_j) \ge (2/5)\mu_i^*$.
    Whenever a bundle is created, its items are removed from consideration.
    Thus, the bundles created must be pairwise disjoint. Moreover, it can easily
    be verified that any bundle created in the algorithm has value at least
    $(2/5)\mu_i^*$. It remains to show that each bundle created is independent
    and that $k$ bundles are created.

    Note that each bundle $B$ created in phase one must be independent, as $B$
    either consists of a single item with non-zero value or is a subset of an
    independent bundle returned by the approximate valuation oracle. Since every
    item $j \in M_H$ with $v_i(\{j\}) \ge (2/5)\mu_i^*$ is allocated in phase
    one, it holds that $\mu_i^*/5 < v_i(\{j\}) < (2/5)\mu_i^*$ in phase two.
    Thus, $v_i(v_i^o(\{j_1, j_2\})) \ge (2/5)\mu_i^*$ requires that
    $v_i^o(\{j_1, j_2\}) = \{j_1, j_2\}$. Hence, every bundle created in phase
    two is also independent.

    Let $r$ be the number of bundles created in phase one. We wish to show that
    the number of bundles created in phase two is at least $k - r$, whenever $r
    < k$. Consider an MMS partition, $P = (P_1, \dots, P_n)$, for agent $i$ and
    let $M'' = M_H \cup M_L$ be the remaining items at the start of phase two.
    Let $P' = (P'_1, \dots, P'_n)$ be such that $P'_j = P_j \cap M''$. We claim
    that there are at least $k - r = n - \ell - r$ bundles $P'_j \in P'$ with
    \[v_i(P'_j) \ge \frac{2}{5}\cdot\frac{\mu_i^*}{(1 - \frac{1}{n + 1})}\eqcomma\]
    $|P'_j \cap M_H| \ge 2$ and $v_i(P'_j \cap M_H) \ge (2/5)\mu_i^*$. Note that
    the last two conditions hold as long as the first condition holds, as
    otherwise there is at most one item $j \in (M_H \cap P'_j)$ that contributes
    to the value of $P'_j$ and $v_i(v_i^o(\{j\} \cup M_L)) \ge (1 -
    1/(n + 1))v_i(P'_j) \ge (2/5)\mu_i^*$. In other words, phase one could not yet
    have completed.

    To see that there exists at least $n - \ell - r$ bundles $P'_j \in P'$ with
    sufficient value, we wish to show that after removing from $P$ the items
    in $M \setminus M''$, at most $\ell + r$ of the bundles no longer have
    sufficient value. In other words, it must be shown that items with a
    combined value of
    \[
        \mu_i - \frac{2}{5} \cdot \frac{\mu_i^*}{1 - \frac{1}{n + 1}}
        \ge \mu_i - \frac{2}{5}\cdot\frac{(1 + \frac{1}{5n - 1})\mu_i}{1 -
        \frac{1}{n + 1}}\eqcomma
    \]
    have been removed from no more than $\ell + r$ of the bundles in $P$.

    First, consider the singleton bundles created in phase one, letting $r_1$
    denote the number of such bundles. As they contain only a single item each,
    they remove value from at most $r_1$ bundles in $P$. In the worst case, each
    one of these $r_1$ bundles no longer have sufficient value. The value of the
    remaining $n - r_1$ bundles has not changed.

    Notice that every other bundle created in phase one has a value of strictly
    less than $(3/5)\mu_i^*$, as otherwise there is some low-valued item in the
    bundle, that could be removed from the bundle without decreasing the value
    of the bundle beyond $(2/5)\mu_i^*$. As all the bundles are independent,
    $(3/5)\mu_i^*$ is also a bound on the total value of the items in each of
    the bundles. Similarly, the total value of the items in each bundle $C_j$ is
    also bound by $(3/5)\mu_i^*$. As in the proof of \cref{thr:existence}, we
    can therefore bound the value of the items removed from the bundles in $P$,
    that were not affected by the singleton bundles, by $(\ell + r -
    r_1)(3/5)\mu_i^*$. Thus, we wish to show that

    \begin{align*}
        (\ell + r - r_1 + 1)\left(\mu_i - \frac{2}{5} \cdot \frac{\left(1 +
        \frac{1}{5n - 1}\right)\mu_i}{1 - \frac{1}{n + 1}}\right) &\ge (l + r -
        r_1)\frac{3}{5}\mu_i^* \\
        (\ell + r - r_1 + 1)\left(1 - \frac{2}{5} \cdot \frac{1 +
        \frac{1}{5n - 1}}{1 - \frac{1}{n + 1}}\right)\mu_i &\ge \frac{3}{5}(l + r -
        r_1)\left(1 + \frac{1}{5n - 1}\right)\mu_i \\
        (\ell + r - r_1 + 1)\left(1 - \frac{2}{5} \cdot \frac{1 + \frac{1}{5n -
        1}}{1 - \frac{1}{n + 1}}\right) &\ge \frac{3}{5}(l + r - r_1)\left(1 +
        \frac{1}{5n - 1}\right) \\
        1 - \frac{2}{5} \cdot \frac{1 + \frac{1}{5n - 1}}{1 - \frac{1}{n + 1}}
        &\ge \frac{3}{5}\left(\frac{l + r - r_1}{l + r - r_1 + 1}\right)\left(1
        + \frac{1}{5n - 1}\right)
    \end{align*}
    Note that $\ell + r - r_1 \le n - 1$, as $r < k$. Thus, we need only show
    that
    \begin{align*}
        1 - \frac{2}{5} \cdot \frac{1 + \frac{1}{5n - 1}}{1 - \frac{1}{n + 1}}
        &\ge \frac{3}{5}\left(\frac{n-1}{n}\right)\left(1 + \frac{1}{5n -
        1}\right) \\
        5n - 2n \cdot \frac{\frac{5n}{5n - 1}}{\frac{n}{n + 1}} &\ge
        3(n-1)\left(\frac{5n}{5n - 1}\right) \\
        5n - \frac{2(5n)(n + 1)}{5n - 1} &\ge \frac{3(n - 1)(5n)}{5n - 1} \\
        5n(5n - 1) - 2(5n)(n + 1) &\ge 3(5n)(n - 1) \\
        5n(5n - 1) &\ge 5(5n)(n - 1) + 4(5n) \\
        25n^2 - 5n &\ge 25n^2 - 5n
    \end{align*}
    As a consequence, there are at least $k - r$ bundles that satisfy the
    requirements at the start of phase two.

    We claim that when there are $x \ge 0$ bundles $P'_j \in P'$ with $|P'_j
    \cap M_H | \ge 2$ and $v_i(P'_j \cap M_H) \ge (2/5)\mu_i^*$, at least $x$
    bundles will be created in phase two. Indeed, for any pair of items
    $j_1,j_2 \in M_H$ with $v_i(\{j_1, j_2\}) \ge (2/5)\mu_i^*$, it holds that
    $v_i(\{j_1\}) > (1/3)v_i(\{j_1, j_2)\})$ and $v_i(\{j_2\}) > (1/3)v_i(\{j_1,
    j_2\})$. Thus, as $n \ge 2$ and $v_i(v_i^o(B)) \ge (1 - 1/(n + 1))v_i(B)$,
    it holds that $v_i^o(\{j_1, j_2\}) = \{j_1, j_2\}$ for at least one pair
    $j_1, j_2 \in P'_j \cap M_H$ for each $P'_j$ that satisfies the two
    criteria.  Since $P$ is a partition, there must exist at least $x$ pairwise
    disjoint pairs of items that are each internally connected by an edge in the
    graph $G$, and at least $x$ bundles are created in phase two.

    Unless $r \ge k$, in which case enough bundles are created in phase one,
    there are at least $k - r$ bundles $P'_j \in P'$ with $|P'_j \cap M_H | \ge
    2$ and $v_i(P'_j \cap M_H) \ge (2/5)\mu_i^*$. As a consequence, at least $k$
    bundles are created across the two phases.
\end{proof}

There is a trade-off in \cref{lem:2/5-partition-bounds} between the
restrictiveness of the error bound on the approximate valuation oracle and the
upper bound on the estimate, $\mu_i^*$, for $\mu_i$. If the error of the oracle
is worse, then the upper bound on $\mu_i^*$ must be stricter, and vice versa.
Note that if there is no error in the estimate $\mu_i^*$, i.e., $\mu_i$ is
known, then the error of the valuation oracle could be as big as $\min\{3/(2n +
3), 1/3\}$.\footnote{The upper bound of $1/3$ is required for the approximate
valuation oracle to discover the pairs with value at least $(2/5)\mu_i$ in the
second phase.} Similarly, if the valuation oracle is exact, it can be shown that
the error bound on $\mu_i^*$ can be as big as $3/(5n - 3)$. Thus, the error
bounds used in \cref{lem:2/5-partition-bounds} differ from the optimum by only a
small constant factor in either direction.

Next, we show that \cref{alg:2/5-poly-partition} runs in polynomial time in the
number of agents and items, when $\mu_i^* > 0$.

\begin{lemma}\label{lem:2/5-poly-partition}
    Given an instance, \cref{alg:2/5-poly-partition} runs in polynomial time in
    the number of agents, $n$, and items, $m$, when $\mu_i^* > 0$.
\end{lemma}

\begin{proof}
    It can easily be verified that every individual operation in phase one can
    be performed in polynomial time. Particularly, every time $v_i$ is queried,
    the supplied bundle is known to be independent or has a cardinality of one.
    Moreover, the for-loop has at most $|M_h| \le m$ iterations. As $\mu_i^* >
    0$, every iteration of the outer while-loop (line~\ref{alg:line:while-1})
    creates a non-empty bundle $B$, removing the items in $B$ from further
    consideration. Thus, there are at most $m$ iterations of the loop. The inner
    loop reduces the size of $B$ in each iteration, and is thus also restricted
    to at most $m$ iterations.  Therefore, the first phase can be performed in
    polynomial time.

    Phase two can also be verified to run in polynomial time. To construct the
    graph $G$, at most $m^2$ pairs of items need to be considered. Moreover, the
    maximum-cardinality matching can be computed in time polynomial to the
    number of vertices ($O(m)$) and edges ($O(m^2)$) in the graph and has
    cardinality at most $m/2$.
\end{proof}

A crucial, missing piece of the $2/5$-approximation algorithm is the computation
of the error bound, $\mu_i^*$, for each agent $i$. As mentioned, it is NP-hard
to compute $\mu_i$ or even approximate it within a factor better than $2/3$ by
\cref{thr:2/3-np-hard}. Instead, we take an approach used in other MMS
approximation algorithm, including the algorithm of \citeauthor{Li:21}. That is,
we start by setting $\mu_i^*$ to some value that is guaranteed to be no less
than $\mu_i$. If the algorithm fails to return a satisfactory allocation, an
agent $i'$ with a too high estimate, $\mu_{i'}^*$, is identified and the value
of $\mu_{i'}^*$ is reduced.

\begin{algorithm}[t]
    \caption{Find a $2/5$-approximate MMS allocation}
    \label{alg:2/5-approximation}
    \begin{algorithmic}[1]
        \REQUIRE Agents $N$, items $M$, and approximate valuation oracles $V$
        \ENSURE A $2/5$-approximate MMS allocation
        \FOR{$i \in N$}
            \STATE Let $\mu_i^* = m \cdot v_i(j)$, where $j$ is the $n$-th most
            valuable item according to $i$
        \ENDFOR
        \WHILE{a (partial) allocation has not been found}
            \STATE Run the lone divider algorithm for the agents in $N$ and items in
            $M$, with estimates $\mu_i^*$ for $\mu_i$. It returns either partial
            allocation $A$ or an agent $i$ which could not produce the required
            bundles when chosen.
            \IF{an agent is returned}
                \STATE $\mu_i^* = \frac{n}{n +
                1}\mu_i^*$\label{alg:line:multiplicative-adjustment}
            \ENDIF
        \ENDWHILE
        \RETURN $A$
    \end{algorithmic}
\end{algorithm}

Observe that \cref{lem:2/5-partition-bounds} allows us to determine some agent
$i$ with a too high estimate. If the lone divider algorithm is ran, with
\cref{alg:2/5-poly-partition} constructing bundles,
\cref{lem:2/5-partition-bounds} guarantees that the algorithm will never fail to
construct bundles for an agent $i$ with a sufficiently low estimate. If the
algorithm fails for an agent $i$, then it must hold that $\mu_i^* > (1 + 1/(n +
1))\mu_i$.  Thus, the estimate of agent $i$ can safely be reduced. Specifically,
if the estimate is lowered by the multiplicative factor $1/(1 + 1/(n + 1)) =
n/(n + 1)$, it still holds that $\mu_i^* > \mu_i$. \Cref{alg:2/5-approximation}
employs this approach, together with an observation of \citeauthor{Li:21} on
upper and lower bounds for the value of $\mu_i$, to find a $2/5$-approximate MMS
allocation.

\begin{lemma}[Claim~3 in~\cite{Li:21}]\label{lem:mms-bounds}
    Let $j$ be the $n$-th most valuable item according to agent $i$. Then,
    $v_i(\{j\}) \le \mu_i \le mv_i(\{j\})$.
\end{lemma}

\begin{theorem}\label{thr:2/5-poly}
    A $2/5$-approximate MMS allocation can be found in polynomial time under
    hereditary set system valuations, given approximate valuation oracles with
    an error bound by $1/(n + 1)$.
\end{theorem}

\begin{proof}
    First, we show that we can restrict our consideration to instances where
    every agent has a non-zero MMS and there are at least two agents. By
    \cref{lem:mms-bounds}, $\mu_i > 0$ if and only if agent $i$ assigns a
    non-zero value to the in her eyes $n$-th most valuable item. Thus, it can be
    determined if $\mu_i > 0$ for each agent in an instance in polynomial time.
    If there are no agents with $\mu_i > 0$, then any allocation can be
    returned. Similarly, if there is only one agent with $\mu_i > 0$, this agent
    can be given every item. Otherwise, every agent with $\mu_i = 0$ can be
    removed to produce a satisfactory instance.

    We wish to show that \cref{alg:2/5-approximation} finds a $2/5$-approximate
    MMS allocation in polynomial time, when there are at least two agents and
    $\mu_i > 0$ for every agent $i \in N$. First, note that every step of the
    lone divider algorithm can be performed in polynomial time. For the creation
    of bundles, this follows from \cref{lem:2/5-poly-partition}. Moreover, as
    every bundle created is independent, the valuation queries needed to
    construct the graph for the maximal-cardinality envy-free matching can also
    be performed in polynomial time. Thus, the lone divider step of
    \cref{alg:2/5-approximation} and subsequently each iteration of the while
    loop runs in polynomial time.

    By \cref{lem:2/5-partition-bounds}, the bundle creation can not fail for an
    agent $i$ if $\mu_i^* \le (1 + 1/(n + 1))\mu_i$, as each bundle $C_j$
    allocated prior to choosing $i$ as the lone divider is independent and has
    value $v_i(C_j) < (2/5)\mu_i^*$. Thus, the number of iterations of the
    while-loop is bound by the number of agents times the maximum number of
    adjustments that can be made to $\mu_i^*$ before $\mu_i^* \le (1 + 1/(n +
    1))\mu_i$ for an agent $i$. By \cref{lem:mms-bounds}, $\mu_i^* \le m\mu_i$
    initially. Thus, the number of adjustments per agent is at most
    \[
        \log_{1 + \frac{1}{n + 1}}m
        = \frac{\ln m}{\ln \left(1 + \frac{1}{n + 1}\right)}
        < \frac{\ln m}{\frac{5}{6(n + 1)}}
        = \frac{6}{5}(n + 1)\ln m\eqcomma
    \]
    where the second to last step follows from the Maclaurin series $\ln(1 + x)
    = \sum_{r = 1}^\infty (-1)^{r + 1} x^r/r$, which can be lower bounded by
    $5/(6(n + 1))$ in the case of $x = \frac{1}{n + 1}$ and $n \ge 2$
    \[\ln\left(1 + \frac{1}{n + 1}\right) = \sum_{r = 1}^\infty (-1)^{r +
    1}\frac{\left(\frac{1}{n + 1}\right)^r}{r} > \frac{1}{n + 1} - \frac{1}{2(n
    + 1)^2} \ge \frac{1}{n + 1} - \frac{1}{6(n + 1)} = \frac{5}{6(n + 1)}\eqdot\]
    Thus, the number of iterations of the while-loop is polynomially bounded and
    the algorithm runs in polynomial time.

    Consider the (partial) allocation return by the algorithm. It holds by
    \cref{thr:sufficient-conditions,lem:2/5-partition-bounds} that each agent
    $i$ is allocated a bundle with value at least $(2/5)\mu_i^*$. Since $\mu_i^*
    > \mu_i$ at the start of the algorithm and $\mu_i^*$ is only adjusted by the
    multiplicative factor $n/(n + 1)$ if $\mu_i^* > (1 + 1/(n + 1))\mu_i$, we
    get that $\mu_i^* > (n/(n + 1))(1 + 1/(n + 1))\mu_i = \mu_i$ holds
    throughout the execution of the algorithm. Consequently, each agent receives
    a bundle with value at least $(2/5)\mu_i$.
\end{proof}

Finally, we note that the approach used to obtain a $2/5$-approximation
algorithm, can be used to obtain weaker approximations when the error of the
inaccurate approximation oracle is larger. For example, in the case in which the
accuracy is bound by a constant. The proof is deferred to
\cref{app:inaccurate-oracle}, due to its similarity to that of
\cref{thr:2/5-poly}.

\begin{theorem}\label{thr:inaccurate-oracle}
    Given an instance with hereditary set system valuations and inaccurate
    approximation oracles with error bound $0 \le \epsilon < 1$, a ($\frac{1 -
    \epsilon}{1 + (3/2)(1 - \epsilon)}$)-approximate MMS allocation can be found
    in polynomial time if
    \begin{enumerate}[label=(\arabic*)]
        \item There exists a $\delta$ polynomial in the size of the
            instance, such that $\epsilon \le 1 - \frac{1}{\delta}$; and
        \item Either $\epsilon \le 1/3$ or the independence of bundles of
            cardinality two can be checked in polynomial time.
    \end{enumerate}
\end{theorem}

\section{Entitled Hereditary Set System Valuations}\label{sec:entitlement}

Hereditary set system valuations require that the same hereditary set system is
used for every agent's valuation function. In this section, we consider a
natural relaxation of this requirement, that encapsulates additional real-world
settings, e.g., budget constraints. Specifically, we consider cases in which
agent's valuations are not required to be based on the same hereditary set
system. Instead, there is a separate hereditary set system $H_i = (M,
\mathcal{F}_i)$ for each agent $i$. The only requirements placed on the choice
of hereditary set systems is that there exists an ordering of the agents, $a_1,
\dots, a_n$, such that $\mathcal{F}_{a_1} \subseteq \mathcal{F}_{a_2} \subseteq
\dots \subseteq \mathcal{F}_{a_n}$. Conceptually, this means that some agents
are allowed to derive value from additional subsets of items. We refer to this
type of valuations by \emph{entitled hereditary set system valuations}, as some
agents are more entitled in their choice of bundles.

We now show that our existence and approximation results from hereditary set
system valuations hold even under entitled hereditary set system valuations. The
extension follows from the simple observation that the bundles allocated in the
lone divider algorithm can be independent for every remaining agent, if whenever
the lone divider $i$ is chosen, $i$ is selected as the remaining agent that
appears at the earliest point in the agent ordering.

\begin{lemma}\label{lem:existence-entitlement}
    A ($n/(2n - 1)$)-approximate MMS allocation always exists under entitled
    hereditary set system valuations.
\end{lemma}

\begin{proof}
    Modify line~\ref{alg:line:divider} of \cref{alg:lone-divider} so that the
    agent chosen is the remaining agent with the most restrictive hereditary set
    system. Then, at any point during the allocation of the algorithm, the
    already allocated bundles are independent for every remaining agent, and the
    proof of \cref{thr:existence} can be repeated with the additional
    observations from \cref{lem:existence}.
\end{proof}

Following the same method for choosing the lone divider, it can easily be shown
that also the approximation results hold if the ordering of the agents is known.
However, the results can be shown to hold even when the ordering is not given
and valuations can only be accessed through an approximate valuation oracle. The
idea is to partially determine the agent order through the error bound of the
valuation oracle. Notice that if for some agent $i'$ and bundle $B$ it holds
that $v_{i'}(v_{i'}^o(B)) < (1 - \epsilon)\sum_{j \in B} v_{i'}(\{j\})$, then
$B$ is not independent for $i'$. Thus, whenever a lone divider has created a set
of bundles, it can be checked if this condition holds for some remaining agent
$i'$ and created bundle $B_j$. In the case that the condition holds for some
agent $i'$, we know that $i'$ precedes $i$ in the ordering, as $B_j$ is
independent for $i$. Agent $i$ can then be replaced by $i'$ as the lone divider.
If the condition does not hold for any agent $i'$, it provides an upper bound on
the combined value of the items in $B$, which can be to satisfy the $\sum_{j'
\in C_j} v_{ij'} \le (3/5)\mu_i^*$ bound of \cref{lem:2/5-partition-bounds}.
These changes to the lone divider algorithm are shown in
\cref{alg:lone-divider-entitled}.

\begin{algorithm}[t]
    \caption{Modified lone divider algorithm for entitled hereditary set system
    valuations}
    \label{alg:lone-divider-entitled}
    \begin{algorithmic}[1]
        \REQUIRE A set of agents, $N$, a set of items, $M$, approximate
        valuation oracles, $v_i^o$, with error bound $\epsilon$, and values
        $x_i$ for $i \in N$.  \ENSURE A (partial) allocation $A$ with $v_i(A_i) \ge x_i$
        \WHILE{$N \neq \emptyset$}
        \STATE select some $i \in N$
        \STATE create $|N|$ pairwise disjoint bundles, $B_1, \dots, B_{|N|}$,
        where every $B_j$ is independent in $H_i$ and has $v_i(B_j) \ge
        x_i$\label{alg:line:divider-entitled}
       \IF{$v_{i'}(v_{i'}^o(B_j)) < (1 - \epsilon)\sum_{j' \in B_j} v_{i'j'}$
        for some $i' \in N$ and $B_j$}
            \STATE let $i = i'$
            \STATE go to line~\ref{alg:line:divider-entitled}
        \ENDIF
        \STATE find a maximal-cardinality envy-free matching, $M_{EF}$, with
        regards to $N$, in the bipartite graph $G = (N \cup (B_1, \dots,
        B_{|N|}), \{(i', B_j) : v_{i'}(v_{i'}^o(B_j)) \ge
        x_{i'} \text { if } i' \neq i \text{, otherwise } v_{i}(B_j) \ge x_i\})$
        \STATE allocate to each matched agent $i$ in $M_{EF}$ their matched
        bundle in $M_{EF}$
        \STATE update $N$ and $M$ by removing matched agents and items from
        matched bundles in $M_{EF}$
        \ENDWHILE
    \end{algorithmic}%
\end{algorithm}

\begin{lemma}\label{lem:poly-entitlement}
    A $2/5$-approximate MMS allocation can be found in polynomial time under
    entitled hereditary set system valuations when the ordering of the agents is
    unknown, given an approximate valuation oracle with error bound by $1/(n +
    1)$.
\end{lemma}

\begin{proof}
    In the same way as for \cref{thr:2/5-poly}, we can assume that $\mu_i > 0$
    for every agent and that there is at least two agents. We wish to show that
    \cref{alg:2/5-approximation}, with the modified lone divider algorithm of
    \cref{alg:lone-divider-entitled}, finds a $2/5$-approximate MMS allocation
    in polynomial time.

    Note that at any stage of \cref{alg:lone-divider-entitled}, we have the
    following bound on the value of the items in any already allocated bundle
    $C_j$ for any remaining agent $i$.
    \[
        \sum_{j' \in C_j} v_{ij'} \le
        \frac{1}{1 - \frac{1}{n + 1}}v_{i}(v_{i}^o(C_j)) <
        \frac{1}{1 - \frac{1}{n + 1}}\cdot\frac{2}{5}\mu_{i}^* \le
        \frac{n + 1}{n}\cdot\frac{2}{5}\mu_{i}^* \le
        \frac{3}{2}\cdot\frac{2}{5}\mu_{i}^* = \frac{3}{5}\mu_{i}^*
    \]
    The second to last step follows from the fact that $n \ge 2$. This bound on
    the item values in $C_j$ is no worse than the $(3/5)\mu_i^*$ bound required
    by \Cref{lem:2/5-partition-bounds}. Thus, following the same steps as in the
    proof for \cref{thr:2/5-poly} yields the conclusion that a $2/5$-approximate
    MMS allocation is returned from \cref{alg:2/5-approximation} using
    \cref{alg:lone-divider-entitled} for the lone divider step.

    It remains to show that the running time of the algorithm is polynomial.
    First, note that in each iteration of \cref{alg:lone-divider-entitled}, the
    lone divider is changed at most $n - 1$ times. Indeed, every time the lone
    divider is changed from some agent $i$ to some agent $i'$, it holds that
    $v_{i'}(v_{i'}^o(B_j)) < (1 - 1/(n + 1))\sum_{j' \in B_j} v_{i'j'}$ for some
    bundle $B_j$ that is independent for agent $i$. By the definition of
    $v_{i'}^o$, $B_j$ can not be independent for $i'$, and thus
    $\mathcal{F}_{i'} \subset \mathcal{F}_{i}$. Since the hereditary set system
    of the lone divider becomes strictly more restrictive each time the lone
    divider is changed, each agent can be assigned as the lone divider at most
    once in each iteration. Thus, \cref{alg:lone-divider-entitled} runs in
    polynomial time under the same arguments as for \cref{alg:lone-divider}.
    Moreover, since \cref{lem:2/5-partition-bounds} holds, the number of
    iterations of the while-loop in \cref{alg:2/5-approximation} can be shown to
    be polynomial in $n$ and $m$ in the same way as in the proof of
    \cref{thr:2/5-poly}.
\end{proof}

A further natural extension of the class of valuation functions is to relax the
correlation requirement between the hereditary set systems. We consider now the
class of valuation functions in which there is no correlation between the
hereditary set systems that the valuation functions are based on. We will refer
to this class of valuations by \emph{asymmetric hereditary set system
valuations}. Note that asymmetric hereditary set system valuations are a proper
subclass of XOS valuations.

Unfortunately, our results do not easily extend to asymmetric hereditary set
system valuations. In fact, it can be shown that for any number of agents $n \ge
2$ and $\epsilon > 0$, there exists an instance with $n$ agents such that there
is no $(1/2 + \epsilon)$-approximate MMS allocation. A stark contrast to
\cref{lem:existence,lem:existence-entitlement}, which guarantee that for any
number of agents $n$ there exists an $\epsilon > 0$, such that $(1/2 +
\epsilon)$-approximate MMS allocations exist for every instance with $n$ agents
under (entitled) hereditary set system valuations. Our nonexistence proof is a
modification of a proof for an equivalent result for general XOS valuations,
given by \citet{Ghodsi:22}. Their proof uses valuation functions that are
almost, but not entirely asymmetric hereditary set system valuations.

\begin{lemma}
    For any $n \ge 2$ and $\epsilon > 0$, there exists an instance with $n$
    agents and asymmetric hereditary set system valuations with binary marginal
    gains, such that no $(1/2 + \epsilon)$-approximate MMS allocation exists.
\end{lemma}

\begin{proof}
    For a given number of agents $n \ge 2$, we construct an instance with $m =
    2n$ items. For each agent $i$ and item $j$, let $v_{ij} = 1$.  Let $H_1 =
    (M, \{S \subset M : |S| = 1 \text{ or } S = \{2i, 2i + 1\}, i \in
    \mathbb{N}\})$.  Similarly, let $H_2 = (M, \{S \subset M : |S| = 1 \text{ or
    } S = \{2i + 1, 2i + 2\}, i \in \mathbb{N} \text{ or } S = \{1, 2n\})$.
    Notice that any agent with valuation function based on $H_1$ or $H_2$ has
    $\mu_i = 2$, as the partitions $(\{1, 2\}, \dots, \{2n - 1, 2n\})$ and
    $(\{2, 3\}, \dots, \{2n - 2, 2n - 1\}, \{1, 2n\})$ contain $n$ bundles
    valued at $2$ for $H_1$ and $H_2$, respectively. Furthermore, since $v_{ij} =
    1$, $\mu_i \le m/n = 2$.

    Let $n - 1$ agents have valuation functions defined by $H_1$ and the
    remaining agent, $i$, have a valuation function defined by $H_2$. Since
    valuations are integer, any ($1/2 + \epsilon$)-approximate MMS allocation
    requires each agent to receive a bundle with value at least $2$. For the
    agents based on $H_1$, this requires a bundle that contains some subset
    $\{1, 2\}, \dots, \{2n - 1, 2n\}$. Thus, providing each of the $n - 1$
    agents based on $H_1$ with value at least $2$, leaves at most two items $2j,
    2j + 1$ for the last agent. Note that $\{2j, 2j + 1\}$ is not independent in
    $H_2$ for any $j$. Thus, $v_i(\{2j, 2j + 1\}) = 1$. Consequently, it is
    impossible to provide every agent with a bundle with value greater than $1$,
    and no ($1/2 + \epsilon$)-approximate MMS allocation can exist.
\end{proof}

\section{Constrained Fair Allocation}\label{sec:constraints}

In this section we highlight how the improved results for hereditary set system
valuations improve existence and approximation results for a range of
constrained fair allocation problems. We are interested in constraints that
dictate, independently of the overall allocation, which bundles an agent may
receive---a bundle that satisfies such a criteria for an agent is referred to as
\emph{feasible} for the agent. Particularly, we are interested in constraints in
which the set of feasible bundles for each agent forms a hereditary set system,
i.e., the types of constraints where every subset of a feasible bundle is
feasible. Moreover, for our results to be applicable we must limit our
consideration further, to constraints in which the set of feasible bundles is
either the same for each agent (symmetric constraints) or where the hereditary
set systems satisfy the same requirements as for entitled hereditary set system
valuations. That is, there is some ordering of the agents, $a_1, \dots, a_n$,
such that whenever $i < j$, every feasible bundle for $a_i$ is also feasible for
$a_j$.

Constraints differ from hereditary set system valuations in a fundamental way.
Under hereditary set system valuations, there is no requirement that each agent
receives an independent bundle. However, when constraints are enforced each
agent must receive a feasible bundle. Fortunately, as the bundles constructed in
our lone divider algorithms, except under entitled hereditary set system
valuations, are independent, each of the earlier results can produce instead a
partial allocation, $A = (A_1, \dots, A_n)$, with each $A_i \in A$ independent
and $v_i(A_i) \ge \alpha\mu_i$. For entitled hereditary set system valuations,
this can also be achieved by slight modification of
\cref{alg:lone-divider-entitled}.  Notice that the edges in the graph $G$ for
the envy-free matching are constructed using the value of the bundle returned by
the approximate valuation oracle of each agent. Thus, if this subset of items is
allocated to the matched agent, rather than bundle created by the lone divider,
the matched agent receives an independent bundle with sufficient value.

\begin{observation}\label{obs:partial}
    The results of
    \cref{thr:existence,thr:2/5-poly,thr:inaccurate-oracle,lem:existence-entitlement,lem:poly-entitlement}
    hold also for partial allocations in which every agent receives an
    independent bundle.
\end{observation}

Note that some constraints require complete allocations. Our results do not
necessarily extend to this requirement, as it might for some constraints not be
possible to produce a complete allocation from a partial allocation without
reallocating already allocated items. For these constraints, a different
approach must be taken. However, if it is always possible to extend a partial
allocation to a complete allocation by simply allocating the unallocated items
in some way, the result are still applicable.  To this point, we note that
matroid constraints, in which the set of feasible bundles is the independent
sets of some given matroid, are often defined to require complete allocations.
Thus, our results on hereditary set system valuations are only applicable to
certain types of matroids that allow any partial feasible allocation to be
extended to a complete allocation, such as partition matroids---or if the
completeness requirement is foregone. Note that under matroid constraints, an
exact valuation oracle can easily be constructed from an independence oracle.

We start by considering conflicting items constraints \cite{Chiarelli:23a}.
Under conflicting items constraints, each instance is furnished by a graph $G =
(M, E)$. A bundle is feasible if it contains no pair of neighbouring items in
the graph $G$, i.e., the set of feasible bundles are the independent sets of
$G$. As any subset of an independent set in a graph is also an independent set,
the feasible bundles form a hereditary set system. Note that allocations are
required to be complete under conflicting items constraints. However, if the
maximum degree, $\Delta(G)$, of any vertex in $G$ is strictly smaller than $n$,
a complete allocation can be obtained from a partial allocation by greedily
allocating any unallocated item to an agent without any conflicting items. Under
this mild assumption, \citet{Hummel:22} showed that $1/3$-approximate MMS
allocations always exist for additive valuations.
\Cref{thr:existence,obs:partial} allow this result to be improved to $1/2$.

\begin{lemma}
    For an instance of fair allocation of conflicting items with additive
    valuations and $n > \Delta(G)$, there exists a complete $1/2$-approximate
    MMS allocations.
\end{lemma}

\begin{proof}
    Let $H = (M, \{S \subseteq M : S \text{ independent in } G\})$ be a
    hereditary set system. For each agent $i \in N$, let $v'_i$ be the valuation
    function given by the item values of $v_i$ and the hereditary set system
    $H$. Consider the fair allocation instance with hereditary set system
    valuations given by $N$, $M$ and the valuation functions $v_i'$. Let $\mu_i$
    and $\mu_i'$ be the MMS of agent $i$ in the original instance and the
    instance with hereditary set system valuations, respectively. Then, $\mu_i'
    \ge \mu_i$, as $v_i'(B) = v_i(B)$ for any feasible bundle $B$.

    By \cref{thr:existence,obs:partial}, there exists a partial
    allocation $A = (A_1, \dots, A_n)$ such that every $A_i$ is independent in
    $H$ and $v_i'(A_i) \ge \mu_i'/2$. Since $A_i$ is independent in $H$, $A_i$
    is also independent in $G$ and it holds that $v_i(A_i) = v_i'(A_i) \ge
    \mu_i'/2 \ge \mu_i/2$. The partial allocation $A$ can be extended to a
    complete allocation by in turn allocating each unallocated item $j$ to some
    agent with no item in conflict with $j$. At least one such agent must exist,
    as there are at most $\Delta(G) < n$ items in conflict with $j$. Thus, a
    complete $1/2$-approximate MMS allocation exists.
\end{proof}

As feasible bundles are independent sets under conflicting items constraints,
approximating $v_i'$ is equivalent to solving the weighted independent set
problem. Thus, our results do not generally improve the existing
$1/\Delta(G)$-approximation result of \citet{Hummel:22}. However, for classes of
graphs in which the weighted independent set problem can be solved in polynomial
time, a $2/5$-approximation result can be shown to hold in the same way as the
existence result.

Next, we consider the case of interval scheduling constraints, introduced by
\citet{Li:21f}. Under interval scheduling constraints, each item $j \in M$ is
endowed by a value $v_{ij}$, a \emph{processing time} (or \emph{duration}) $p_j
\in \mathbb{N}^+$, a \emph{release time} $r_j \in \mathbb{N}^+$ and a
\emph{deadline} $d_j \in \mathbb{N}^+$, where $d_j \ge r_j + p_j - 1$. A set of
items $S$ is \emph{feasible}, if it is possible to schedule, without overlap,
each item $j \in S$ to $p_j$ consecutive time periods in $[r_j, d_j]$, where a
time period is given by $[t, t + 1)$ for any $t \in \mathbb{N}^+$. That is,
$v_i(B)$ is given by
\[v_i(B) = \max_{\text{ feasible } B' \subseteq B}\sum_{j \in B'} v_{ij}\eqdot\]
It can easily be verified that the valuation functions under interval scheduling
constraints are hereditary set system valuations, as if a set $S$ is feasible,
then any subset $S' \subseteq S$ must also be feasible by simply leaving the
items in $S \setminus S'$ out of the schedule. We note that \citeauthor{Li:21f}
concentrated on finding partial allocations in which each bundle is feasible,
referred to as \emph{feasible allocations} (or \emph{feasible schedules}).

Making use of \cref{thr:existence,thr:inaccurate-oracle}, we can improve both
the existence and approximation guarantees of \citeauthor{Li:21f} ($1/3$ and
$0.24 - \epsilon$, respectively). Their approximation result is dependent on how
well the valuation functions can be approximated in polynomial time. However,
their approximation guarantee is $\beta/(2 + \beta)$, as opposed to the
$\beta/(1 + (3/2)\beta)$ guarantee of \cref{thr:inaccurate-oracle}, where
$\beta$ is the approximation ratio for the valuation function,. Since $\beta \le
1$, the latter provides a better approximation for any $\beta > 0$. Note that
\citeauthor{Li:21f} base their $0.24$ result on $\beta = 0.6448$, attributed to
a result of \citet{Im:20}. This seems to be an oversight by the authors, as the
$0.6448$ result of \citeauthor{Im:20} is for the unweighted case of the related
machine scheduling problem, while the valuation functions are equivalent to the
weighted case. Thus, our result instead makes use of $\beta = 1/2$, from a
result of \citet{Bar-Noy:99}.

\begin{lemma}
    For an instance of fair allocation with interval scheduling constraints,
    there always exists a feasible $1/2$-approximate MMS allocation and a
    feasible $2/7$-approximate MMS allocation can be found in polynomial time.
\end{lemma}

\begin{proof}
    Each instance has hereditary set system valuations, as for any set of
    feasible items $S$, any subset is also feasible. Thus, the existence follows
    directly from \cref{thr:existence,obs:partial}. The approximation
    result follows from \cref{thr:inaccurate-oracle,obs:partial}, as a
    approximate valuation oracle with error bound $1/2$ can be constructed from
    a result of \citet{Bar-Noy:99}. Note that as $\epsilon = 1/2 > 1/3$,
    \cref{thr:inaccurate-oracle} requires that the independence of bundles with
    cardinality two can be queried in polynomial time. This independence is
    trivial to determine in polynomial time, by checking the two different
    orderings of the items in the schedule.
\end{proof}

Another studied type of constraints for which our results provides improvements
to the state-of-the-art is budget constraints \cite{Wu:21}. Under budget
constraints, each item $j \in M$ has an attached \emph{size} (or \emph{cost}),
$s_j \in \mathbb{R}_{\ge 0}$. Each agent $i$, is endowed with a budget $b_i \in
\mathbb{R}_{\ge 0}$ and is only allowed to receive a bundle $B$ with $\sum_{j
\in B} s_j \le b_i$. The set of feasible bundles for each agent $i$ forms a
hereditary set system $H_i = (M, \mathcal{F}_i)$, as for any bundle $B$ within
the budget of agent $i$, any bundle $B' \subset B$ is also within $i$'s budget.
Moreover, for any pair of agents $i, j$ with $b_i \le b_j$, it holds that
$\mathcal{F}_i \subseteq \mathcal{F}_j$. Since there is no requirement on
complete allocations, the results from \cref{sec:entitlement} apply to budget
constraints. Specifically, we can improve the $1/3$ existence guarantee and $1/3
- \epsilon$ (for any $\epsilon > 0$) approximation guarantee of \citet{Li:23a}
to $1/2$ and $2/5$, respectively.

\begin{lemma}
    For an instance of fair allocation under budget constraints, a
    $1/2$-approximate MMS allocation always exists, and a $2/5$-approximate MMS
    allocation can be found in polynomial time.
\end{lemma}

\begin{proof}
    For each agent $i \in N$, let $H_i = (M, \{S \subseteq M : \sum_{j \in S}
    s_j \le b_i\})$ be a hereditary set system. Moreover, let $v'_i$ be the
    valuation function given by the item values of $v_i$ and the hereditary set
    system $H_i$. Consider the fair allocation instance given by $N$, $M$ and
    the valuation functions $v_i'$. This instance has entitled hereditary set
    system valuations, as for any pair of agents $i, j$ with $b_i \le b_j$,
    every feasible bundle for $i$ is also feasible for $j$.

    Let $\mu_i$ and $\mu_i'$ be the MMS of agent $i$ in the original instance
    and the instance with entitled hereditary set system valuations,
    respectively. Then, $\mu_i' = \mu_i$, as $v_i'(B) = v_i(B)$ for any feasible
    bundle $B$. Consequently, the existence result follows directly from
    \cref{lem:existence-entitlement,obs:partial}.

    Determining the exact value of $v_i'$ for any bundle $B$ is NP-hard, as it
    is equivalent to solving the knapsack problem. However, there exists a FPTAS
    for the knapsack problem, which yields an inaccurate valuation oracle with
    error bound by $\epsilon = 1/(n + 1)$ and running time polynomial in
    $1/\epsilon = n + 1$. Thus, the approximation result follows directly from
    \cref{lem:poly-entitlement,obs:partial}.
\end{proof}

\section{Conclusion}

The lone divider approach has allowed for improvements to both the existence
guarantee and polynomial-time computability of $\alpha$-approximate MMS
allocations under hereditary set system valuations, along with several
constrained fair allocation problems. However, several open problems remain for
hereditary set system valuations. First and foremost, for any $n > 2$, there
remains a gap between the lower and upper bound for guaranteed existence of
$\alpha$-approximate MMS allocations. Second, there is a gap between the $1/2$
existence guarantee and the $2/5$ guarantee of the approximation algorithm.

For the second problem, we note that it is possible to generalise the two-phase
approach of \cref{alg:2/5-poly-partition} to work for any $\alpha \in (2/5,
1/2)$. Specifically, for some $k \ge 3$, the three following changes allow the
algorithm to construct bundles with $v_i(B_j) \ge \alpha\mu_i$ for $\alpha =
k/(2k + 1)$:
\begin{enumerate}[label=(\arabic*)]
    \item Lower the requirement for an item to be high-valued from $\mu_i/5$ to
        $\mu_i/(2k + 1)$.
    \item In phase one, consider each subset $S \subseteq M_H$ with $|S| < k$
        and $v_i(S) \le \alpha\mu_i$, instead of only individual items $j \in
        M_H$, when greedily constructing bundles.
    \item In phase two, find a maximum-cardinality collection of pairwise disjoint
        independent bundles $S \subseteq M_H$ with $|S| \le k$ and $v_i(S) \ge
        \alpha\mu_i$.
\end{enumerate}
The correctness of the modifications follows from the same argument as for
$\alpha = 2/5$ ($k = 2$). In the first phase, the value of a bundle will not
exceed $\alpha\mu_i + \mu_i/(2k + 1) = (1 - \alpha)\mu_i$. Thus, when the second
phase starts, there are still enough bundles in the MMS partition of $i$ with a
remaining value of at least $\alpha\mu_i$ to guarantee a sufficiently high
cardinality for the collection of pairwise disjoint bundles constructed in the
phase.

Notice that for some fixed $k$, the running time of phase one is still
polynomial. However, finding a maximum-cardinality collection in phase two is
hard, due to its close similarity to the two NP-hard problems of hypergraph
matching and set packing. Note that, using algorithms for set packing, it can be
shown that for any fixed $k$, the modified algorithm has a running time that is
FPT parameterised by $n$.

\bibliographystyle{plainnat}
\bibliography{paper}

\appendix

\section{Proof of Theorem~\ref{thr:inaccurate-oracle}}
\label{app:inaccurate-oracle}

We prove \cref{thr:inaccurate-oracle} using a similar approach as for
\cref{thr:2/5-poly}. Specifically, consider a variant of
\cref{alg:2/5-poly-partition} (see \cref{alg:alpha-poly-partition}) in which
every occurrence of the factor $2/5$ has been replaced by $\alpha = (1 -
\epsilon)/(1 + (3/2)(1 - \epsilon))$, where $\epsilon$ is the error bound of the
inaccurate valuation oracle. We wish to show that this algorithm constructs a
sufficient number of satisfactory disjoint independent bundles under similar
conditions as in \cref{lem:2/5-partition-bounds}. Then, an equivalent approach
to \cref{alg:2/5-approximation} can be taken to obtain the result. The proof of
\cref{lem:alpha-poly-partition} follows the same steps as the proof of
\cref{lem:2/5-partition-bounds}.

\begin{lemma}\label{lem:alpha-poly-partition}
    Consider an instance given by a set $N$ of $n \ge 2$ agents, a set $M$ of
    items and hereditary set system valuations $V$. Let $\mu^*_i > 0$ be an
    estimate of the MMS of agent $i \in N$, with
    \[\mu^*_i \le (1 + \frac{3 - 3\epsilon}{5n - 3n\epsilon + 3\epsilon +
    3})\mu_i\eqcomma\]
    $v^o_i$ an approximate valuation oracle for $v_i$, with error bound by
    $0 \le \epsilon < 1$, and \[\alpha = \frac{1 - \epsilon}{1 + \frac{3}{2}(1 -
    \epsilon)}\eqdot\]
    Given $0 \le \ell < n$ pairwise disjoint bundles $C_1, \dots,
    C_\ell$ with $\sum_{j' \in C_j} v_{ij'} \le (3/2)\alpha\mu^*_i$,
    \cref{alg:alpha-poly-partition} will return $k = n - \ell$ pairwise disjoint
    independent bundles $B_1, \dots, B_k$ with $v_i(B_j) \ge \alpha\mu^*_i$,
    when ran for $v^o_i$, $M' = M \setminus (\bigcup_{1 \le j \le \ell}C_j)$,
    $k$, $\alpha$ and $\mu^*_i$ if $\epsilon \le 1/3$ or $v_i(v^o_i(B)) =
    v_i(B)$ for bundles $B$ with $|B| = 2$.
\end{lemma}

\begin{proof}
    First, note that every bundle created in the algorithm is independent and
    has a value of at least $\alpha\mu_i^*$ by the same argument as in
    \cref{lem:2/5-partition-bounds}. Moreover, the bundles are pairwise disjoint
    as their items are removed from consideration before proceeding with the
    construction of the next bundle in phase one and the matching upholds this
    guarantee in phase two.

    Let $r$ be the number of bundles created in phase one. We wish to show that
    the number of bundles created in phase two is at least $k - r$, whenever $r
    < k$. Consider an MMS partition, $P = (P_1, \dots, P_n)$, for agent $i$ and
    let $M'' = M_H \cup M_L$ be the remaining items at the start of phase two.
    Let $P' = (P_1',\dots,P_n')$ be such that $P_j' = P_j \cap M''$. We claim
    that there are at least $k - r = n - \ell - r$ bundles $P_j' \in P'$ with
    \[
        v_i(P'_j) \ge \alpha\frac{\mu_i^*}{1 - \epsilon}\eqcomma
    \]
    $|P'_j \cap M_H| \ge 2$ and $v_i(P'_j \cap M_H) \ge \alpha\mu_i^*$. Note
    that the last two conditions hold as long as the first condition holds, as
    otherwise there is at most one item $j \in (M_H \cap P'_j)$ that contributes
    to the value of $P'_j$ and $v_i(v_i^o(\{j\} \cup M_L)) \ge (1 -
    \epsilon)v_i(P'_j) \ge \alpha\mu_i^*$. In other words, phase one could not
    yet have completed.

    To see that there exists at least $n - \ell - r$ bundles $P'_j \in P'$ with
    sufficient value, we wish to show that after removing from $P$ the items in
    $M \setminus M''$, at most $\ell + r$ of the bundles no longer have
    sufficient value. In other words, it must be shown that items with a
    combined value of
    \[
        \mu_i - \alpha\frac{\mu_i^*}{1 - \epsilon} \ge \mu_i - \alpha\frac{1 +
        \frac{3 - 3\epsilon}{5n - 3n\epsilon + 3\epsilon + 3}}{1 -
        \epsilon}\mu_i
    \]
    have been removed from no more than $\ell + r$ of the bundles in $P$.

    First, consider the singleton bundles created in phase one, letting $r_1$
    denote the number of such bundles. As in the proof of
    \cref{lem:2/5-partition-bounds}, there are at least $n - r_1$ bundles in $P$
    that do not decrease in value after removing only the items in the singleton
    bundles.

    Every other bundle created in phase one has a value of strictly less than
    $(3/2)\alpha\mu_i^*$. By the same argument as in the proof of
    \cref{lem:2/5-partition-bounds}, there is otherwise some low-valued item in
    the bundle that may be removed without reducing the value of the bundle
    beyond $\alpha\mu_i^*$. Combined with the bound of $(3/2)\alpha\mu_i^*$ on
    the combined value of items in each $C_j$, items with a combined value of at
    most $(\ell + r - r_1)(3/2)\alpha\mu_i^*$ have been removed from the $n -
    r_1$ bundles in $P$ not affected by the singleton bundles. Thus, we wish to
    show that
    \begin{align*}
        (\ell + r - r_1 + 1)\left(\mu_i - \alpha\frac{\left(1 +
        \frac{3 - 3\epsilon}{5n-3n\epsilon + 3\epsilon + 3}\right)\mu_i}{1 - \epsilon}\right) &\ge (l + r -
        r_1)\frac{3}{2}\alpha\left(1 + \frac{3 - 3\epsilon}{5n-3n\epsilon + 3\epsilon + 3}\right)\mu_i \\
        (\ell + r - r_1 + 1)\left(1 - \alpha\frac{5n - 3n\epsilon +
        6}{(5n-3n\epsilon + 3\epsilon + 3)(1 - \epsilon)}\right) &\ge (l + r -
        r_1)\frac{3}{2}\alpha\left(\frac{5n - 3n\epsilon + 6}{5n-3n\epsilon + 3\epsilon + 3}\right) \\
        2 - 2\alpha\frac{5n - 3n\epsilon + 6}{(5n-3n\epsilon + 3\epsilon + 3)(1 -
        \epsilon)} &\ge 3\alpha\left(\frac{\ell + r - r_1}{\ell + r - r_1 +
        1}\right)\left(\frac{5n - 3n\epsilon + 6}{5n-3n\epsilon + 3\epsilon +
        3}\right)
    \end{align*}
    Note that $\ell + r - r_1 \le n - 1$, as $r < k$. Thus, we need only show
    that
    \begin{align*}
        2 - 2\alpha\frac{5n - 3n\epsilon + 6}{(5n-3n\epsilon + 3\epsilon + 3)(1 -
        \epsilon)} &\ge 3\alpha\left(\frac{n - 1}{n}\right)\left(\frac{5n -
        3n\epsilon + 6}{5n-3n\epsilon + 3\epsilon +
        3}\right) \\
        2n(5n - 3n\epsilon + 3 + 3\epsilon) - 2n\alpha\frac{5n - 3n\epsilon +
        6}{1 - \epsilon} &\ge 3\alpha(n - 1)(5n - 3n\epsilon + 6)
        \\
        2n(5n - 3n\epsilon + 3 + 3\epsilon) - 2n\left(\frac{5n - 3n\epsilon +
        6}{1 + \frac{3}{2}(1 - \epsilon)}\right) &\ge 3\left(\frac{1 -
        \epsilon}{1 + \frac{3}{2}(1 - \epsilon)}\right)(n - 1)(5n - 3n\epsilon +
        6)
        \\
        2n(1 + \frac{3}{2}(1 - \epsilon))(5n - 3n\epsilon + 3 + 3\epsilon) -
        2n(5n - 3n\epsilon + 6) &\ge 3(1 - \epsilon)(n - 1)(5n - 3n\epsilon + 6)
        \\
        (2n + 3n(1 - \epsilon))(5n - 3n\epsilon + 3 + 3\epsilon) -
        2n(5n - 3n\epsilon + 6) &\ge 3(1 - \epsilon)(n - 1)(5n - 3n\epsilon + 6)
        \\
        \begin{aligned}
            (2n + 3n(1 - \epsilon))((5n - &3n\epsilon + 6) + (3\epsilon - 3)) \\&-
        2n(5n - 3n\epsilon + 6)
        \end{aligned}&\ge
        (3n - 3)(1 - \epsilon)(5n - 3n\epsilon + 6)\\
        (2n + 3n(1 - \epsilon))(3\epsilon - 3)&\ge
        -3(1 - \epsilon)(5n - 3n\epsilon + 6)\\
        24n\epsilon - 15n - 9n\epsilon^2 &\ge
        24n\epsilon - 15n - 9n\epsilon^2 + 18\epsilon - 18 \\
        18 &\ge 18\epsilon
    \end{align*}
    Where the final equation holds as $\epsilon < 1$. As a consequence,
    there are at least $k - r$ bundles that satisfy the requirements at the
    start of phase two.

    As in the proof of \cref{lem:2/5-partition-bounds}, we claim that if there
    are $x \ge 0$ bundles $P'_j \in P'$ with $|P'_j \cap M_H| \ge 2$ and $v_i(P'_j
    \cap M_H) \ge \alpha\mu_i^*$, at least $x$ bundles will be created in phase
    two. Whenever $\epsilon \le 1/3$, it holds for any pair of items $j_1,j_2
    \in M_H$ with $v_i(\{j_1, j_2\}) \ge \alpha\mu_i^*$ that $v_i(\{j_1\}) >
    (1/3)v_i(\{j_1, j_2\})$ and $v_i(\{j_2\}) > (1/3)v_i(\{j_1, j_2\})$, as
    $(\alpha/2)\mu_i^* < v_i(\{j\}) < \alpha\mu_i^*$. Thus, as
    $v_i(v_o(B)) \ge (2/3)v_i(B)$, we have that $(v_o(\{j_1, j_2\}) = \{j_1,
    j_2\}$ whenever $v_i(\{j_1, j_2\}) \ge \alpha\mu_i^*$. Note that if
    $\epsilon > 1/3$, this is guaranteed by the assumption that $v_i(v_o(B)) =
    v_i(B)$ for any bundle $B$ with $|B| = 2$. Consequently, there is in either
    case at least one pair $j_1, j_2 \in P_j' \cap M_H$ for each $P_j'$ that
    satisfies the two criteria. Since $P$ is a partition, there must exist at
    least $x$ pairwise disjoint pairs of items that are each internally
    connected by an edge in the graph $G$, and at least $x$ bundles are created
    in phase two.

    Unless $r \ge k$, in which case enough bundles are created in phase one,
    there are at least $k - r$ bundles $P_j' \in P'$ with $|P'_j \cap M_H| \ge
    2$ and $v_i(P'_j \cap M_H) \ge \alpha\mu_i^*$. As a consequence, at least
    $k$ bundles are created across the two phases.
\end{proof}

\begin{algorithm}[t]
    \caption{Find $k$ disjoint independent bundles with value at least
    $\alpha\mu_i^*$ for valuation function $v_i$}
    \label{alg:alpha-poly-partition}
    \begin{algorithmic}[1]
        \REQUIRE An approximate valuation oracle $v^o_i$, items $M'$, a number
        $k$, an $\alpha$ and an estimate $\mu_i^*$
        \ENSURE $k$ pairwise disjoint independent bundles $B_1, \dots, B_k$ with
        $v_i(B_j) \ge \alpha\mu_i^*$ or \textsc{Failure}
        \STATE $M_H = \{j \in M' : v_i(\{j\}) > \frac{1}{2}\alpha\mu_i^*\}$
        \STATE $M_L = M' \setminus M_H$
        \COMMENT{Phase one}
        \FOR{$j \in M_H$ with $v_i(\{j\}) \ge \alpha\mu_i^*$}
            \STATE create bundle $\{j\}$
            \STATE $M_H = M_H \setminus \{j\}$
        \ENDFOR
        \WHILE{$\exists j \in M_H$ with $v_i(v^o_i(\{j\} \cup M_L)) \ge
        \alpha\mu_i^*$}
            \STATE $B = v^o_i(\{j\} \cup M_L)$
            \WHILE{$\exists j' \in B$ with $v_i(B \setminus \{j'\}) \ge
            \alpha\mu_i^*$}
                \STATE $B = B \setminus \{j'\}$
            \ENDWHILE
            \STATE create bundle $B$
            \STATE $M_H = M_H \setminus B$
            \STATE $M_L = M_L \setminus B$
        \ENDWHILE
        \COMMENT{Phase two}
        \STATE $G = (M_H, \{\{j_1, j_2\} \in M_H \times M_H : v_i(v^o_i(\{j_1,
        j_2\})) \ge \alpha\mu_i^*\})$
        \STATE Find a maximum-cardinality matching $M_G$ in $G$
        \FOR{$\{j_1, j_2\} \in M_G$}
            \STATE create bundle $\{j_1, j_2\}$
        \ENDFOR
        \IF{less than $k$ bundles have been created}
            \RETURN \textsc{Failure}
        \ELSE
            \RETURN $k$ created bundles
        \ENDIF
    \end{algorithmic}
\end{algorithm}

Notice that the running time of the bundle creation algorithm is not affected by
the change from $2/5$ to $\alpha$. Thus, it remains to show that the
readjustment method used in \cref{alg:2/5-approximation} can be used also here
with a polynomial bound on the number of readjustments.

\begin{proof}[Proof of \cref{thr:inaccurate-oracle}]
    Consider a variant of \cref{alg:2/5-approximation} in which the
    multiplicative factor for $\mu_i^*$ on
    line~\ref{alg:line:multiplicative-adjustment} has been exchanged for
    \[\frac{1}{1 + \frac{3 - 3\epsilon}{5n - 3n\epsilon + 3\epsilon +
    3}}\]
    Note that when $n \ge 2$,
    \[0 < \frac{3 - 3\epsilon}{5n + 3} < \frac{3 - 3\epsilon}{5n - 3n\epsilon +
    3\epsilon + 3} = \frac{3(1 - \epsilon)}{5n - 3(n - 1)\epsilon + 3} <
    \frac{3}{2n + 3} \le \frac{3}{7}\]
    Thus, letting $x = \frac{3 - 3\epsilon}{5n - 3n\epsilon + 3\epsilon + 3}$,
    we get that for $r \ge 1$
    \[\frac{x^r}{r} - \frac{x^{r + 1}}{r + 1} > \frac{x^r - xx^r}{r} >
    \frac{4}{7}\cdot\frac{x^r}{r} > 0 \]
    Consequently, it holds that
    \[
        \ln (1 + x) = \sum_{r = 1}^\infty (-1)^{r + 1}x^r/r
        > x - \frac{x^2}{2}
        > \frac{4}{7}x
        \ge \frac{4}{7}\cdot\frac{3 - 3\epsilon}{5n - 3n\epsilon + 3\epsilon + 3}
        > \frac{4}{7}\cdot\frac{3}{\delta(5n + 3)}
        > \frac{1}{\delta(5n + 3)}\eqcomma
    \]
    where the first step is due to the Maclaurin series $\ln(1 + x) = \sum_{r =
    1}^\infty (-1)^{r + 1}x^r/r$, and the second to last step from the
    assumption that $\epsilon \le 1 - 1/\delta$ for some $\delta$ that is
    polynomial in the size of the instance.

    By this observation, the correctness of the theorem follows directly from
    \cref{lem:alpha-poly-partition,lem:2/5-poly-partition} using the exact same
    proof as for \cref{thr:2/5-poly}. Particularly, the maximum number of
    adjustments to $\mu_i^*$ per agent $i$ is polynomial in the size of the
    input, as both $\delta$, $n$ and $m$ are polynomial in the size of the
    input, and we have that
    \[
        \log_{1 + \frac{3 - 3\epsilon}{5n - 3n\epsilon + 3\epsilon + 3}} m
        = \frac{\ln m}{\ln \left(1 + \frac{3 - 3\epsilon}{5n - 3n\epsilon +
        3\epsilon + 3}\right)}
        < \frac{\ln m}{\frac{1}{\delta(5n + 3)}}
        < \delta(5n + 3)\ln m\qedhere
    \]
\end{proof}

\end{document}